\documentclass[floatfix,superscriptaddress, aps,pra,reprint,10pt,nofootinbib]{revtex4-2}

\usepackage[utf8]{inputenc}
\usepackage{amsmath}
\usepackage{amsfonts}
\usepackage{dsfont}
\usepackage{amssymb}
\usepackage{amsthm}
\usepackage{graphicx}
\usepackage{caption}
\usepackage{braket}
\usepackage{comment}
\usepackage{enumitem}
\usepackage{microtype}

\usepackage{hyperref}
\usepackage[svgnames,dvipsnames]{xcolor}
\hypersetup{colorlinks,linkcolor={DarkBlue},citecolor={DarkBlue},urlcolor={DarkBlue}}

\usepackage[most]{tcolorbox}

\usepackage[all]{hypcap}
\usepackage{textcomp, complexity}
\usepackage{tabularx,environ}
\usepackage{physics}
\usepackage{mathpazo}
\usepackage[normalem]{ulem}
\usepackage[capitalise,compress]{cleveref}
\crefname{section}{Sec.}{Secs.}
\Crefname{section}{Section}{Sections}
\crefname{section}{Sec.}{Subsections}
\Crefname{subsection}{Section}{Subsections}

\newtheorem{thm}{Theorem}
\newtheorem{lemma}[thm]{Lemma}

\newtheorem{conj}[thm]{Conjecture}
\theoremstyle{remark}

\theoremstyle{definition}

\makeatletter

\newcolumntype{\expand}{}
\long\@namedef{NC@rewrite@\string\expand}{\expandafter\NC@find}

\NewEnviron{problem}[2][]{%
  \def\problem@arg{#1}%
  \def\problem@framed{framed}%
  \def\problem@lined{lined}%
  \def\problem@doublelined{doublelined}%
  \ifx\problem@arg\@empty%
    \def\problem@hline{}%
  \else%
    \ifx\problem@arg\problem@doublelined%
      \def\problem@hline{\hline\hline}%
    \else%
      \def\problem@hline{\hline}%
    \fi%
  \fi%
  \ifx\problem@arg\problem@framed%
    \def\problem@tablelayout{|>{\bfseries}lX|c}%
    \def\problem@title{\multicolumn{2}{|l|}{%
        \raisebox{-\fboxsep}{\textsc{#2}}%
      }}%
  \else
    \def\problem@tablelayout{>{\bfseries}lXc}%
    \def\problem@title{\multicolumn{2}{l}{%
        \raisebox{-\fboxsep}{\textsc{#2}}%
      }}%
  \fi%
\par
\noindent%
  \begin{tabularx}{\columnwidth}{\expand\problem@tablelayout}%
    \problem@hline%
    \problem@title\\[\fboxsep]%
    \BODY 
  \end{tabularx}%
    \\
}
\makeatother
\newcounter{parentnumber}

\definecolor{applegreen}{rgb}{0.55,.71, 0}
\colorlet{dominik}{applegreen!80!black}

\newclass{\postBQP}{PostBQP}
\newclass{\sP}{\#\P}
\newclass{\postNGBS}{PostNGBS}
\newclass{\FBPP}{FBPP}

\begin{document}

\title{Quantum Computational Advantage via High-Dimensional Gaussian Boson Sampling}

\newcommand{\fu}{Dahlem Center for Complex Quantum Systems, Freie Universit{\"a}t Berlin, 14195 Berlin, Germany}
\newcommand{\mfu}{Department of Mathematics and Computer Science, Freie Universit{\"a}t Berlin, 14195 Berlin, Germany}
\newcommand{\hzb}{Helmholtz-Zentrum Berlin f{\"u}r Materialien und Energie, 14109 Berlin, Germany}
\newcommand{\iqim}{Institute for Quantum Information and Matter, Caltech, Pasadena, CA 91125, USA}
\newcommand{\xa}{Xanadu, Toronto, ON M5G 2C8, Canada}
\newcommand{\jqi}{Joint Quantum Institute, NIST/University of Maryland, College Park, MD 20742, USA}
\newcommand{\quics}{Joint Center for Quantum Information and Computer Science, NIST/University of Maryland, College Park, MD 20742, USA}
\newcommand{\ulm}{Institut f{\"u}r Theoretische Physik and Center for Integrated Quantum Science and Technology (IQST),\\
Albert-Einstein-Allee 11, Universit{\"a}t Ulm, 89069 Ulm, Germany}
\newcommand{\montr}{Department of Engineering Physics, \'Ecole Polytechnique de Montr\'eal, Montr\'eal, QC, H3T 1JK, Canada}

\author{Abhinav Deshpande}
\thanks{These two authors contributed equally}
\affiliation{\quics}
\affiliation{\jqi}
\affiliation{\iqim}
\author{Arthur Mehta}
\thanks{These two authors contributed equally}
\affiliation{\xa}
\affiliation{Department of Mathematics, University of Toronto, Toronto, ON M5S 1A1 Canada}
\author{Trevor Vincent}
\affiliation{\xa}
\author{Nicol\'as Quesada}
\affiliation{\xa}
\affiliation{\montr}
\author{Marcel Hinsche}
\affiliation{\fu}
\author{Marios Ioannou}
\affiliation{\fu}
\author{Lars Madsen}
\affiliation{\xa}
\author{Jonathan Lavoie}
\affiliation{\xa}
\author{Haoyu Qi}
\affiliation{\xa}
\author{Jens Eisert}
\affiliation{\fu}
\affiliation{\hzb}
\affiliation{\mfu}
\author{Dominik Hangleiter}
\affiliation{\quics}
\affiliation{\fu}
\author{Bill Fefferman}
\affiliation{Department of Computer Science, The University of Chicago, Chicago, IL 60637, United States}
\author{Ish Dhand}
\thanks{ishdhand@gmail.com}
\affiliation{\ulm}

\begin{abstract} 
Photonics is a promising platform for demonstrating a quantum computational advantage (QCA) by outperforming the most powerful classical supercomputers on a well-defined computational task.
Despite this promise, existing proposals and demonstrations face challenges.
Experimentally, current implementations of Gaussian boson sampling (GBS) lack programmability or have prohibitive loss rates.
Theoretically, there is a comparative lack of rigorous evidence for the classical hardness of GBS.
In this work, we make progress in improving both the theoretical evidence and experimental prospects.
We provide evidence for the hardness of GBS, comparable to the strongest theoretical proposals for QCA.
We also propose a new QCA architecture we call high-dimensional GBS, which is programmable and can be implemented with low loss using few optical components.
We show that particular algorithms for simulating GBS are outperformed by high-dimensional GBS experiments at modest system sizes.
This work thus opens the path to demonstrating QCA with programmable photonic processors.
\end{abstract}

\maketitle
\section*{Introduction}
We are arriving at an exciting era for quantum computing in which quantum experiments are pushing the limits of what is efficiently computable by the most powerful classical supercomputers.
The first major goal for this era is the demonstration of a scalable quantum advantage or \textit{quantum computational advantage (QCA)} (also termed ``{quantum computational supremacy}'') over classical computers. 
QCA is important as a probe of the foundations of computer science, where it can be seen as an experimental violation of the extended Church-Turing thesis, and it also serves as an important benchmarking tool for comparing near-term experiments on different platforms in a fair and consistent manner.
The recent groundbreaking demonstrations of QCA~\cite{arute2019quantum,Zhong2020} constitute the first significant experimental evidence against the extended Church-Turing thesis.

Notwithstanding, multiple potential loopholes have been pointed out~\cite{pednault2019leveraging,kalai2014gaussian,renema2020marginal}.
Indeed, QCA will not be marked by a single isolated experiment but rather will be established by gradually improving and scaling up ``high complexity'' experiments run over the course of many years, which improving classical algorithms will try to simulate.
Our confidence that we have arrived in this new era will grow as multiple experiments, performed in different physical architectures, independently reach this conclusion in a comparable fashion. 
In this way, the goal may be seen as being analogous to Bell inequality violations, which were originally conducted in landmark experiments starting in the 1970s performed on a variety of different platforms but only much later were loopholes closed.

In the same vein, theoretical results about QCA justify the classical hardness of simulating an experiment in the realm of asymptotically large system sizes.
In order to interpret conclusions from experiments performed at a fixed system size, we should also consider the concrete cost of simulating these finite-size experiments using known algorithms.
The two lines of inquiry are complementary to each other and support each other in a claim that any experiment is likely impossible to feasibly simulate with current hardware.

Among different approaches to demonstrating QCA~\cite{Boixo2018,arute2019quantum,Aaronson2013,Zhong2020}, photonics provides a promising path as it enables room-temperature operation, fast gate speeds and remarkable potential for scalability~\cite{bourassa2021blueprint,bartolucci2021fusion}.
Arguably, the most feasible approach to demonstrating QCA with photonics is to perform the \emph{Gaussian boson sampling (GBS)} protocol~\cite{hamilton2017gaussian,Kruse2019}.
Indeed, this protocol is at the heart of the recent QCA demonstration performed by a team from USTC~\cite{Zhong2020}, which employed a GBS device with 100 modes and an average of around 45 photons.
However, GBS has several important limitations.
On the experimental side, current implementations of GBS either lack programmability~\cite{Zhong2020} or have high loss rates, which could render the system classically simulable~\cite{Garcia-Patron2019,Qi2020}.
Also, from a theoretical standpoint, there is a comparative lack of complexity-theoretic evidence for the hardness of GBS~\cite{renema2020marginal} and an understanding of the classical runtime of concrete algorithms to simulate GBS instances. 

In this work, we aim to address these challenges.
We close important theoretical loopholes in the hardness argument for GBS and provide evidence for the hardness of classically simulating GBS even in the presence of loss. 
We moreover propose a new, programmable architecture for GBS that promises better robustness to loss in a near-term experiment and an asymptotic quantum speedup over classical algorithms.
In addition, our proposed architecture is designed so that it is outside of known regimes where current algorithms can simulate finite-size GBS instances in feasible time, as we show through numerical benchmarking.

We first address the open theoretical questions about GBS, namely, hardness in the regime with little overall noise in the form of optical loss.
More specifically, to provide complexity-theoretic evidence for the hardness of approximately simulating GBS, we prove average-case hardness of computing output probabilities in the noise-free case, formulate the so-called `hiding property'~\cite{Aaronson2013} for GBS in terms of a random-matrix theory conjecture and provide analytical and numerical evidence for this conjecture. 
These results bring GBS to the level of evidence shared by other QCA proposals such as \emph{random circuit sampling} (RCS) and conventional \emph{boson sampling} (see e.g., Refs.~\cite{Aaronson2013,Boixo2018,Bouland2019}) up to a mild conjecture in random matrix theory.
We then show that average-case hardness of computing output probabilities still holds in a regime of high loss rates, building on recent results~\cite{bouland2021noise}, and discuss the implications of this result on the noise-regimes in which one may still expect GBS to be hard to simulate on a classical computer.
These results bolster the evidence for QCA in the USTC experiment and also any future GBS experiments.

Given these theoretical results, we then address the programmability versus low-loss tradeoff in current architectures. 
To this end, we introduce a new architecture, high-dimensional GBS, using a time-domain approach. 
This architecture can be implemented programmably with low overall loss while at the same time being hard to simulate for the known classical simulation algorithms.
The hardness of this architecture is borne out by the hardness of computing output probabilities for the lossy, high-dimensional GBS setup.
These results provide evidence of classical hardness for asymptotic system sizes.
In the realm of finite system sizes, we take care to avoid regimes where the experiment can be tractably simulated~\cite{Garcia-Patron2019,oh2021classical}, such as when the linear-optical network has limited connectivity (such as one-dimensional network topology) or when the system is too lossy.
Our proposed high-dimensional GBS architecture voids these algorithms by taking advantage of the enhanced connectivity available in higher dimensions than one.
In this realm, efficient algorithms can be successful in a variety of regimes~\cite{Garcia-Patron2019,oh2021classical} such as when the linear-optical network has limited connectivity (such as one-dimensional network topology) or when the system is too lossy.

To this end, we perform benchmarking simulations to estimate the cost of high-dimensional GBS against state-of-the-art algorithms for simulating GBS and for simulating high-dimensional quantum many-body systems~\cite{quesada2020exact,gray2020hyper}.
These simulations give evidence that classically intractable instances of high-dimensional GBS can be built in the lab with a small number of optical components.
These advantages make high-dimensional GBS an ideal near-term architecture for demonstrating QCA with a programmable photonic device.

Thus, by addressing the above-mentioned shortcomings of GBS from the theoretical and experimental perspectives and understanding the limits of its classical simulability through both asymptotic analysis and finite-size benchmarking, this work paves the way toward more `loophole-free' demonstrations of QCA with a programmable photonic quantum device.

\subsection*{Hardness of approximate GBS}
\label{sec:hardness}

We begin by reviewing and strengthening the hardness argument for the task of simulating GBS as introduced in Refs.~\cite{hamilton2017gaussian,Kruse2019}.
We first introduce the model of Gaussian boson sampling and then examine the evidence for the hardness of approximate boson sampling.
Two properties are required for establishing complexity-theoretic hardness of sampling using the standard QCA arguments, namely \emph{hiding} and \emph{average-case hardness of approximating probabilities}. 
Here, we strengthen the results of Refs.~\cite{hamilton2017gaussian,Kruse2019} by providing strong evidence for these properties in GBS. 
Specifically, we reduce the hiding property to a highly plausible conjecture in random matrix theory, for which we provide analytical and numerical evidence. 
Additionally, we provide evidence for approximate average-case hardness by proving approximate worst-case hardness and near-exact average-case hardness of computing the output probabilities.
Thereby, up to a random-matrix-theory conjecture, we bring the hardness argument for GBS to the same standard as that of boson sampling.
We then extend the latter results to the case of computing output probabilities of noisy GBS, which can be well-motivated when the noise model describing the experimental data is trusted.
These results show that the evidence of a quantum ``signal'' remains in the output distribution even in the presence of noise.
Finally, we discuss the implications of these results on the complexity of simulating GBS in the presence of noise. 

\subsection*{Recap: Gaussian boson sampling}
\label{Sec:GBSRecap}

GBS is the computational task of sampling the photon number statistics of a Gaussian state. 
Obtaining a sample from a typical GBS experiment involves the following steps.
First, a general Gaussian state is prepared at the input, often taken to be $M$ single-mode squeezed vacuum states.
These states are then interfered on an $M$-mode linear/optical interferometer containing beam-splitters and phase shifters.
Finally, the Gaussian state at the output of the interferometer is impinged on $M$ \emph{photon-number-resolving (PNR)} detectors.
The resulting pattern of photon number outcomes from the detectors is the required sample.
Because single-mode squeezed states can be generated and interfered deterministically at room temperatures with high rates, GBS is experimentally feasible on large scales already today, as evidenced by the recent experiment from USTC~\cite{Zhong2020}.

In more detail, a typical GBS experiment involves interfering $M$ single-mode squeezed vacuum states with squeezing parameters $\{r_i\}_{i=1}^M$ at an interferometer specified by an $M \times M$ linear-optical unitary matrix $U$.
Note that some of the modes can be optionally prepared in the vacuum state, and these can be specified by setting their squeezing parameter to zero.

The probability of detecting $n_1$ photons in the first mode, $n_2$ in the second, and so on, denoted by $\mathbf n = (n_1,\ldots, n_M)$, is
\[\label{eq:gbsprobs}
\Pr(\mathbf{n}) = \frac{\left|\mathrm{Haf}(A_{\mathbf{n},\mathbf{n}})\right|^2}{\prod_{j=1}^M n_j! \cosh r_j}.
\]
Here, $A = A^T= U \left(\oplus_{i=1}^M \tanh(r_i) \right) U^T$ is the so-called adjacency matrix of the (pure, zero-displacement) Gaussian state~\cite{hamilton2017gaussian}, and $A_{\mathbf{n},\mathbf{n}}$ is the symmetric matrix of size $N=\sum_{i=1}^M n_i$ (i.e.\ the total photon number) obtained by repeating the $i^{\mathrm{th}}$ column and row of $A$ a total of $n_i$ times.
In particular, if $n_i=0$ then the corresponding row and column is deleted. Finally, the Hafnian $\mathrm{Haf}(\cdot)$ of a symmetric $N \times N$ matrix $B$ is given by
\[\label{eq_Hafnian}
\mathrm{Haf}(B) = \sum_{\mu \in \mathrm{PMP}(N)} \prod_{(i,j) \in \mu} B_{i,j},
\]
where $\mathrm{PMP}(N)$ is the set of perfect matching permutations of $N$ elements for even $N$, i.e., permutations $\mu: [N] \to [N]$ satisfying $\mu(2k-1)< \mu(2k)$, $\mu(2k-1)< \mu(2k+1)$. Equivalently, this is the set of all $N!/(2^{N/2} (N/2)!) = (N-1)!!$ ways of partitioning the set $\{1,2,\ldots,N \}$ into $N/2$ subsets of size 2.
The Hafnian of a $0\times 0$ matrix is defined to be $1$ and that of an odd-size matrix is defined to be $0$, which is a manifestation of the fact that squeezed states are supported on even photon number states only.
By allowing for arbitrary linear-optical unitaries and arbitrary squeezing parameters on each squeezer, an arbitrary symmetric matrix $A$ can be encoded (up to scaling pre-factors) into a Gaussian state.
For generic instances, the best-known algorithms to calculate Hafnians have a runtime scaling as $N^3 2^{N/2}$ where $N$ is the size of the matrix~\cite{bjorklund2019faster}.

\subsection*{Recap: Approximate sampling hardness of boson sampling}
\label{sec_surveyhardness}

Before we state our technical results, we review the main steps of the hardness argument for conventional boson sampling as given by Aaronson and Arkhipov~\cite{Aaronson2013}.
These steps provide context for the hardness results of GBS that we present below.

In a standard boson sampling experiment, instead of interfering single-mode squeezed states at an interferometer as done in Gaussian boson sampling, an $N$-photon $M$-mode Fock state is prepared and evolved under a linear-optical unitary and then measured in the photon-number basis. The boson sampling task is to, given a linear-optical unitary as an input, output samples from the output distribution of a corresponding boson sampling experiment.

Aaronson and Arkhipov showed that it is not possible for a classical computer to efficiently do this task unless certain complexity-theoretic conjectures are false. 
In particular, they reduced the task of approximating the probabilities of outputs to the task of efficient sampling, making use of an approximate counting algorithm due to Stockmeyer~\cite{stockmeyer_complexity_1983}.
This probability estimation can in turn be related to approximating the permanent of a certain sub-matrix of the linear-optical unitary, which is provably hard for a class known as $\#\mathsf{P}$ \cite{Valiant1979}.  While the Stockmeyer reduction is not efficient, the existence of a classical efficient sampling algorithm would imply that $\#\mathsf{P}$-hard problems could be solved using fewer computational resources than expected, amounting to an argument by contradiction.

The main difficulty in the hardness argument for boson sampling arises when extending it to the setting of \emph{approximate sampling}. 
Here, the task is to sample from from any distribution that is within constant-size total-variation distance from a given ideal boson sampling distribution.
This additional constraint takes into account that actual devices are bound to achieve only some finite and typically additive precision. 
In this setting, one may therefore argue for a separation of computational power between quantum and classical devices. 

Given this constraint, the hardness argument for the task of approximate sampling must take into account that the constant error budget on the distribution can be distributed \emph{arbitrarily} across all outcome probabilities. 
In particular, this means that any specific outcome probability of the actually sampled distribution might have a large (constant-size) error when compared to the ideal distribution, which would imply that the sampler cannot be used to estimate the true outcome probabilities. 
To get around this issue, the argument is extended to random problem instances: via a property of the distribution over problem instances called \emph{hiding}, one can then translate typical outcomes of fixed instances to fixed outcomes of random instances.
This enforces that with high probability, the overall constant error budget for the entire distribution is manifest in small errors on the individual probabilities that are proportional to the inverse size of the sample space, that is, $\propto 1/\binom{M}{N}$.
Technically, in standard boson sampling, showing the hiding property boils down to showing that the distribution of any small enough sub-matrix of a Haar-random unitary is approximately (in total-variation distance) an entry-wise complex normal distribution.
This implies that all collision-free outcomes are (approximately) equally distributed. 
In particular, Aaronson and Arkhipov show that when $M \in \omega(N^{5})$, we can ``hide'' a random Gaussian matrix in a small enough sub-matrix of the large Haar-random unitary by an appropriate procedure \cite{Aaronson2013} because all of these sub-matrices are indistinguishable from random Gaussian matrices.

For the approximate sampling task to remain computationally intractable, it remains to show that estimating the outcome probabilities up to inverse-exponentially small error is $\#\mathsf{P}$-hard for any large-enough fraction of the problem instances---a property called \emph{approximate average-case hardness}.
More precisely, given a random problem instance, approximating the probability of a given outcome must be $\#\mathsf{P}$-hard with high probability.
As evidence toward this property, it has been shown that \emph{exactly} computing those output probabilities is in fact $\#\mathsf{P}$-hard on average (and this was a motivation for boson sampling in the first place), and it is known that estimating them to the required robustness level is worst-case hard.
However, the hardness of computing those probabilities to a sufficiently large robustness level on average is still unknown.

We now state our results concerning the hardness of general GBS, followed by our proposal for an architecture to perform high-dimensional GBS.

\section*{Results}
\subsection*{Hiding for arbitrarily many squeezers in GBS } \label{sec_hiding}

As mentioned above, the property of hiding in boson sampling can be translated into a property of the distribution of sub-matrices of random linear-optical unitaries chosen from some distribution.
We will now show that a similar property about the distributions of sub-matrices occurring in the evaluation of outcome probabilities also holds in GBS, provided a plausible random-matrix theory conjecture holds.
We focus on the paradigmatic setting in which the linear-optical unitary is drawn from the Haar measure, and we fix the input state to be such that the first $K$ out of $M$ modes are prepared in single-mode squeezed states with identical squeezing parameter $r$, and the remaining $M-K$ modes are prepared in the vacuum state.
Furthermore, we restrict to \emph{collision-free} outcomes $\mathbf{n}$ for which $n_i \in \{0,1\}$, giving rise to a total photon number $N = \sum_{j=1}^M n_j$.
The probability of obtaining such an outcome $\mathbf{n}$ can be written as
\[ \label{eq:GBS_probs_cf}
 \Pr(\mathbf{n}) = \frac{ \tanh^{N} (r)}{\cosh^K (r)} \left|\mathrm{Haf}\left[ \left(U I_K U^T\right)_{\mathbf{n},\mathbf{n} } \right] \right|^2.
\]
Here $I_K = \mathds{1}_K \oplus 0_{M-K}$ denotes the matrix where $\mathds{1}_K$ is a $K$-dimensional identity matrix, $0_{M-K}$ is an $M-K$-dimensional all-zero matrix, and as before, the notation $A_{\mathbf{n}}$ stands for the sub-matrix of $A$ corresponding to the entries of $\mathbf{n}$ (see below Eq.~(1)).
The task of estimating output probabilities of GBS hence corresponds to estimating ${|\mathrm{Haf}( (U I_K U^T)_{\mathbf{n},\mathbf{n}})|}^2$.

To show the GBS hiding property, we need to characterize the distribution of matrices $(U I_K U^T)_{\mathbf{n},\mathbf{n}}$---of which the Hafnian is computed---as induced by the Haar-random choice of $U$ and depending on the scaling relations between $K,N,M$.
To ensure that for every choice of $K$ we can restrict to collision-free outcomes, we choose the squeezing parameter $r$ such that the average photon number $\mathbb{E} [N] = K \cdot \sinh^2 r  \in o(\sqrt M)$~\cite{Aaronson2013}.
This condition ensures that the collision-free outcomes dominate the probability weight. 

Here, we formulate the hiding property in GBS in terms of random matrix theory and provide strong numerical and analytical evidence that it holds regardless of the fraction of squeezed input modes so long as the collision-free condition is satisfied.
Observe that the matrix $(U I_K U^T)_{\mathbf{n},\mathbf{n}} = U_{\mathbf{n},1_K} U_{\mathbf{n},1_K}^T$ can be expressed in terms of the sub-matrix $U_{\mathbf{n} ,1_K} $ of $U$ obtained by choosing rows according to $\mathbf{n}$ and the first $K$ columns. 
To show the hiding property, we need to relate this distribution over matrices to the distribution of the symmetric product $XX^T$ of a complex Gaussian $N \times K$ matrix $X$ with mean $0$ and variance $1/M$, denoted as $X \sim \mathcal G_{N,K}(0,1/M)$.
We provide analytical and numerical evidence for the conjecture that these distributions are indistinguishable for any number of squeezers $K$ satisfying $N \leq K \leq M$. 

\begin{conj}[Hiding in GBS (informal)]
\label{conj:hiding informal}
For any $K$ such that $N \leq K \leq M$ and $N \in o(\sqrt{M})$, the distribution of the symmetric product $U_{\mathbf n, 1_K}U^T_{\mathbf n, 1_k}$ of sub-matrices of a Haar-random $U \in U(M)$ closely approximates the distribution of the symmetric product $XX^T$ of a Gaussian matrix $X \sim \mathcal G_{N,K}(0,1/M)$ in total-variation distance. 
\end{conj}

We provide a formal statement of the conjecture in the Supplementary Material. 
There, we also discuss regimes in which the conjecture is known to be partially true \cite{Aaronson2013,jiang_entries_2009} and provide numerical evidence for it.
Proving this conjecture is an open research problem in random matrix theory.

Conjecture 1 characterizes the distribution of the symmetric product of $N\times K$ sub-matrices of Haar-random unitaries. 
In turn, the Hafnian of such symmetric products determines the output distribution of GBS. 
While in standard boson sampling, the hiding property amounts to hiding a small $N \times N$ Gaussian matrix in a large $M \times M$ Haar-random unitary matrix, in GBS it amounts to hiding a small $N\times N$ symmetric Gaussian matrix $XX^T$ in a large symmetric unitary matrix $U I_K U^T$ for any $K \geq N$.
This means that any particular sub-matrix cannot be distinguished from any other such sub-matrix of the same size, enforcing the constant error budget to be roughly equally distributed across all outcomes.

In particular, the conjecture implies that the hiding property can be achieved with any number $K$ of input squeezers as long as the average total photon number is sufficiently small. 
In turn, the average total photon number is determined by the total amount of squeezing across all input squeezers.
Intuitively, this is due to the fact that the output of a Haar-random unitary does not depend on any fixed input state.
In fact, the average output state is a product of identical thermal states whose average photon number is determined by the total input squeezing.
Importantly, however, the number $K$ is still crucial for the estimation task as it determines the rank of the matrix $(U I_K U^T)_{\mathbf n , \mathbf n}$. 
Since the complexity of computing the Hafnian of a matrix depends on the rank of that matrix~\cite{bjorklund2019faster}, $K$ should be chosen such that it is at least $N$.
Note that the USTC experiment~\cite{Zhong2020} used $K = M/2$ many squeezers, so our results are directly applicable there, strengthening the arguments for their QCA demonstration.

More generally, we consider three regimes of interest, and provide evidence for Conjecture 1 in the Supplementary Material. 
First, the highly sparse regime in which the total number of modes scales as $M= \omega(K^5)$ and the number of photons is equal to the number of squeezers, $N=K$, features provable hiding results due to Ref.~\cite{Aaronson2013}.
Realistic experiments and proposals today operate in the regime $K = c M$, meaning that a constant fraction $c$ of the input modes is squeezed.
In this regime, the result of Ref.~\cite{jiang_entries_2009} provides analytical evidence for hiding in the asymptotic limit as long as the input squeezing is such that $N \in o(\sqrt{M}/\log M)$.
Lastly, we also consider the intermediate regime of how $M$ scales with $K$ between these two extremes, and give numerical evidence for hiding in this general case.

Let us note that we do not expect Conjecture~1 to hold for large $N \in \omega(\sqrt{M})$. Indeed, in this case it is known that hiding fails for standard boson sampling~\cite{jiang_how_2006,Jiang2017}. 

\subsection*{Average-case hardness of computing GBS probabilities} \label{sec_avgcaseGBS}
As outlined earlier, the question of hardness of approximate sampling boils down to whether it is $\#\mathsf{P}$-hard to approximate most output probabilities.
We now show the average-case hardness of this task when the allowed additive approximation error is exponentially small, using techniques from Ref.~\cite{bouland2021noise}.

We have established that the output probabilities of GBS are given in terms of $\left|\mathrm{Haf}( (U I_K U^T)_{\mathbf{n},\mathbf{n} } )  \right|^2$. By virtue of the previous discussion and more precisely, Conjecture~1, the distribution over the $N\times N$ matrices $(U I_K U^T)_{\mathbf{n},\mathbf{n} }$ for Haar random $U$ is well approximated by complex, symmetric Gaussian  matrices $XX^T$. Hence, to show the average-case hardness of computing output probabilities of GBS, it suffices to consider the following problem:

\begin{problem}{$(\delta,\epsilon)$-Squared-Hafnians-of-Gaussians}
Input & A matrix $XX^T$ with $X \sim \mathcal G_{N,K}(0,1/M)$.
\\
Output & ${|\mathrm{Haf}{(XX^T)}|}^2$ to additive error $\epsilon$, with probability $\geq \delta$ over the distribution $\mathcal G_{N,K}(0,1/M)$.
\end{problem}

To complete the argument that an efficient classical approximate sampling algorithm for GBS cannot exist, it remains to prove the \#P-hardness of \textsc{$(\delta,\epsilon)$-Squared-Hafnians-of-Gaussians} as formalized by the following approximate average-case hardness conjecture.

\begin{conj} \label{conj_outputprob}
The \textsc{$(\delta,\epsilon)$-Squared-Hafnians-of-Gaussians} problem is $\# \mathsf{P}$-hard for any
$\epsilon=O\left({N! \tanh^{N}(r)}/{(\cosh^K(r) M^N)}\right)$
and any constant $\delta > 3/4$.
\end{conj}
A proof of Conjectures 1 and 2 would imply that approximate sampling from a random, general GBS instance, is hard on average.
Let us see how.
Assume that there exists a classically efficient sampler $O$ that samples from a associated distribution whose output probability for outcome $i$ is given by $q_i$.
From the promise that this distribution is $\varepsilon$-close in total variation distance to the target distribution, we have $\sum_i |{p_i - q_i}| \leq 2\varepsilon$, where $p_i$ is the corresponding output probability of the target distribution.
Choose a photon number $N$ so that Conjecture 1 is satisfied.
Among the space of all outcomes with $N$ total photons, for a randomly chosen outcome $i$, we have:
\begin{align}
\Pr_i \left[|{p_i - q_i}| \leq \frac{2\varepsilon k}{\binom{M+N-1}{N}}\right] \geq 1 - \frac{1}{k}.
\end{align}
Assuming Conjectures 1 and 2, with probability at least 3/4, $p_i$ is $\# \mathsf{P}$-hard to compute to additive error $\epsilon' = O\left( \frac{N!}{M^N} \right)$.
Therefore, with probability at least $3/4 (1-1/k)$, it is also $\# \mathsf{P}$-hard to compute $q_i$ to within error $\epsilon' + \frac{2\varepsilon k}{\binom{M+N-1}{N}} = O\left(\exp[-N\log N - \Omega(N)]\right)$ assuming $M=\Theta(N^2)$.
On the flip side, the Stockmeyer algorithm \cite{stockmeyer_complexity_1983} allows us to compute the output probability of an arbitrary outcome $q_i$ to within inverse-multiplicative polynomial precision.
Further, by the Markov inequality, most outcomes $q_i$ cannot be much larger than $1/\binom{M+N-1}{N}$:
\begin{align}
\Pr_i \left[ q_i > \frac{l}{\binom{M+N-1}{N}} \right] \leq \frac{\Pr(N)}{l} \leq \frac{1}{l},
\end{align}
where the quantity $\Pr(N)$ is the probability of seeing $N$ total photons.
This means that with probability at least $1-1/l$, $q_i$ can be computed to additive error $O\left(l \exp[-N\log N - \Omega(N)]\right)$ using a $\mathsf{BPP}^{\mathsf{NP}^O}$ machine running the Stockmeyer algorithm.
Therefore, setting $l = 4k$, a $\mathsf{PH}$ algorithm can solve with high probability a problem that is average-case $\#\mathsf{P}$-hard.
This collapses the polynomial hierarchy.

Note that since we have phrased Conjecture 2 in terms of additive error instead of multiplicative error, we do not explicitly need an anticoncentration condition of the form  $\Pr_X\left[p_0 \geq \binom{M+N-1}{N}^{-1} \right]$ $\geq \gamma$ for some constant $\gamma > 0$, as is often conjectured for permanents \cite{Aaronson2013}.
Nevertheless, it is possible that Conjecture 2 already implies a weak form of  anticoncentration.
Informally, an anticoncentration condition states that on a large fraction of the instances the output probabilities are large enough so that a trivial algorithm for computing the probabilities that outputs ``0'' is not sufficient to solve the \textsc{$(\delta,\epsilon)$-Squared-Hafnians-of-Gaussians} problem.
This is because in order for Conjecture 2 to be true, it is necessary for the trivial algorithm to fail with high probability.

As in all other known proposals for demonstrating QCA, this approximate average-case hardness conjecture remains open.
Nonetheless, just like in other proposals, it turns out that one can give evidence for Conjecture 2.
Namely, we can prove a weaker version of the conjecture with a smaller robustness level $\epsilon = O\left(\exp[-6N\log N - \Omega(N)]\right)$ as opposed to $\epsilon = O\left(\exp[-N\log N - \Omega(N)]\right)$ in Conjecture 2.

\begin{thm}
\label{lem:exact average case hardness}
The \textsc{$(\delta,\epsilon)$-Squared-Hafnians-of-Gaussians} problem is $\#\mathsf{P}$-hard under $\mathsf{PH}$ reductions for any $\epsilon \leq O\left(\exp[-6N\log N - \Omega(N)]\right)$ and any constant $\delta > 3/4$.
\end{thm}

We provide a detailed proof of Lemma~3 in the Supplementary Material.
The technique we employ in the proof is a worst-to-average-case reduction~(see, e.g. \cite{Aaronson2013}). That is, by assuming access to an oracle  for the \textsc{$(\delta,\epsilon)$-Squared-Hafnians-of-Gaussians} problem, we show that one in fact approximate $\mathrm{Haf}( XX^T )$ for any matrix $X \in \mathbb{C}^{N \times K}$. This latter task is $\#\mathsf{P}$-hard in the worst-case as we show in the Supplementary Material.
At a high level, the worst-to-average-case reduction relies on the fact that ${|\mathrm{Haf}{(XX^T)}|}^2$ is a low degree (of degree $2N$) polynomial over the entries of the matrix $X$. This allows us to use the oracle to perform polynomial interpolation.
Therefore, by combining this observation with the techniques of Refs.~\cite{Aaronson2013, bouland2021noise,kondo2021fine-grained}, we obtain a worst-to-average case reduction for exactly computing the output probabilities.

Together, our results on the hiding property and the approximate average-case conjecture in GBS, strengthen the evidence for the hardness of approximately simulating GBS in terms of the total-variation distance to the ideal output distribution. 
Given our results, GBS is now on par with the other leading QCA proposals in terms of complexity-theoretic evidence for approximate sampling hardness \cite{Aaronson2013,Boixo2018,Bouland2019,Movassagh2019,bouland2021noise,kondo2021fine-grained}, up to a plausible conjecture in random matrix theory---for which we provided theoretical and numerical evidence. 
To achieve a demonstration covered by those complexity-theoretic results, however, the loss rate at every element of the linear-optical circuit, must scale inversely with the total number of such elements---a daunting challenge from an experimental perspective.

\subsection*{Hardness of computation of output probabilities for noisy GBS} \label{Sec:noisy}

We now go one step further and assess how the complexity-theoretic argument for sampling hardness is affected by more realistic noise levels, in particular, in terms of photon loss.
In terms of scaling, any constant loss rate of the individual optical elements can lead to the output distribution rapidly approaching a classical distribution.
We now show that, nonetheless and surprisingly, an evidence of a quantum signal remains even in the presence of significant loss.
We then discuss to what extent and in which regimes such a quantum signal might lead to the hardness of simulating a lossy GBS experiment.

One of our main results is the average-case hardness of computing the noisy output probability of a random GBS instance, which we obtain by using similar arguments to recent work of Bouland \emph{et al}.~\cite{bouland2021noise}, but now extended to the GBS setting.
Our results are valid for any noise model that is local, stochastic, and is error-detectable using linear optics.
More specifically, we consider a setting where the noise acts locally after every gate, and is of the form
\begin{align}
\mathcal{N}_i[\rho] = (1-\eta_i)\rho + \eta_i \mathcal{E}_i[\rho], \label{eq_noisechannel}                                                                                        \end{align}
where stochasticity requires $\mathcal{E}_i$ to itself be a valid channel (i.e.\ a completely positive trace preserving map) with no identity component.

Consider the following problem.
\begin{problem}{$(\epsilon,\eta)$-NoisyGBS-Probability}
Input & A noisy GBS instance, consisting of the linear-optical unitary $U$ on $M$ modes chosen from the Haar measure $\mathcal{H}$, the squeezing parameters at the input, a description of the noise channels with parameters $\eta_i$, and a description of a collision-free outcome $\mathbf
{n}$ with $N=\mathsf{poly}(M)$ total photons.
Let $\eta=\max_i \eta_i$.\\
Output & With probability $\delta$ over instances, an estimate of the quantity $\Pr(\mathbf{n})$ to additive error $\epsilon$, where $\Pr(\mathbf{n})$ is the probability of obtaining outcome $\mathbf{n}$.\\
& With probability $1-\delta$, an arbitrary output.
\end{problem}
In the above definition, we take $\delta=1$ to mean the worst-case problem.
We prove the following statement of average-case hardness of computing noisy probabilities.
\begin{thm} \label{lem_avgcasehardnoisy}
There exists a noise threshold $\eta_*$ and a sufficiently large polynomial such that the problem \textsc{$(\epsilon,\eta)$-NoisyGBS-Probability} is $\#\mathsf{P}$-hard under $\mathsf{PH}$ reductions for any constant $\delta > 3/4$, $\eta \leq \eta_*$, and $\epsilon \leq 2^{-\mathsf{poly}}$.
\end{thm}
There are two parts to the proof.
The first part is a proof of worst-case hardness of the problem (when $\delta=1$), and the second a worst-to-average-case equivalence.
For worst-case hardness, it turns out that due to a result of Fujii~\cite{fujii2016noise}, it suffices for the noise channel to be a convex combination of the lossless and lossy channels, and to be able to error-detect it.
These conditions are both met for optical loss, since it is a convex combination of the channels corresponding to no photon loss, single-photon loss, and so on~\cite{Oszmaniec2018}.
Moreover, optical loss can also be detected and corrected using only linear-optical operations and photo-detection with high thresholds~\cite{bartolucci2021fusion}.
In Fujii's argument, one postselects on the error-free outcome of an error-detection code and obtains noiseless universal gates for the class of postselected quantum computation, $\mathsf{postBQP}$.
This argument can apply to the optical case as well, since linear optics with postselection is universal for quantum computing \cite{Knill2002a}.

For the worst-to-average-case equivalence, all we need is for the polynomial structure in the problem to be preserved.
This can be satisfied for any local noise model.
Preserving the polynomial structure of the output probability enables us to continue to use the same proof techniques as earlier.

Before moving on, we again remind the reader that we considered the hardness of computing output probabilities.
While these are not tasks that are feasible for any realistic quantum device, our results nevertheless indicate that there is a computationally intractable (but exponentially small) ``quantum signal'' present in the system.

\subsection*{The complexity of noisy and approximate GBS}
\label{sec:noisy apx gbs}

We now discuss the implications of the hardness result for computing noisy GBS probabilities on the complexity of sampling from the output distribution of noisy GBS. 
An immediate implication of this result is that it is classically hard to \emph{exactly sample} from the noisy distribution of a worst-case GBS experiment.
This is because the quantum signal is still present in the distribution, so the argument based on Stockmeyer's algorithm is valid.  
Thus, in the idealized situation in which loss is the only source of noise of an experimental system and the exact loss rate is known, simulating a worst-case GBS experiment is classically intractable. 
Note that loss rates can be inferred from standard optical tomography procedures such as that of Ref.~\cite{rahimi-keshari2011quantum}.
Given that this result links the hardness of simulating the noisy experiment to an exponentially small quantum signal in the form of output probabilities, it is crucial that the noise model accurately captures the working of the device.

We remark that an alternative proof establishing the classical hardness of exact sampling could possibly be made using a postselection argument similar to the one outline in Section 4.2 of Ref. \cite{Aaronson2013}.
As noted in Ref.~\cite{Aaronson2013} however, this approach has not been shown to provide a viable path towards the goal of showing hardness of approximate sampling.
By establishing the average case hardness of approximating output probabilities, Theorem \ref{lem_avgcasehardnoisy}, takes a substantive step towards establishing the hardness of approximate sampling, even in the presence of noise.

We now discuss the more realistic situation in which loss is the predominant, but not the sole, source of noise in a photonic experimental system.
What can we say about the hardness of approximate sampling in such a situation?
To begin with, let us draw on some intuition from RCS schemes acting on $n$ qubits. 
Here, the additive error incurred in estimating output probabilities using the Stockmeyer algorithm is $O(2^{-n})$ with high probability (since this is the size of a typical output probability in an RCS experiment).
In the presence of uncorrected noise, an error of $O(2^{-n})$ in the noisy output probability can be too large for hardness.
For example, there is evidence that with gate-wise depolarizing noise, the probabilities will deviate from uniform by merely $O(2^{-m})$, where $m$ (typically $\omega(n)$) is the total number of gates \cite{Boixo2018}.
This means that approximate-sampling hardness cannot be shown using these techniques, since it is not hard to approximate the noisy probabilities any more.
Indeed, in this regime, the noisy distribution is exponentially close in total-variation distance to the uniform distribution, rendering the approximate sampling task for the noisy distribution classically simulable.

In the case of noisy GBS, the dominant noise model, namely loss, leads to the vacuum state for a sufficiently deep network, which is again a distribution that is easy to classically sample from (similar to the uniform distribution in qubit RCS schemes).
However, if we post-select on a certain minimal number of photons surviving, the distribution need not be easy to simulate.
This post-selection is efficient when the depth of the circuit scales poly-logarithmically in the number of modes.
In this case, the quantum signal will be large enough so that even with an inverse exponential error, deviations from the easy distribution can be detected. 

This excludes the simulation algorithm that samples from an easy-to-simulate distribution such as the one uniform on every photon number sector with every sector sampled according to the ideal photon number distribution.
Ruling out trivial algorithms is a necessary condition for approximate average-case hardness to hold.
In summary, our results indicate that there might be `room in the middle' in terms of gate depth and noise rates, where hardness of sampling might hold.
In fact, this intuition lies at the heart of the high-dimensional architecture (presented below).
This architecture is designed in such a way that only as few gate applications as necessary for hardness are executed, so that the leeway for noise to ruin the hardness of sampling is minimized.
We stress, however, that at the moment, existing proof techniques do not suffice to make a claim of this nature. In fact, in certain regimes of noisy GBS, approximate sampling is known to be classically efficient~\cite{Qi2020}.

\subsection*{High-dimensional GBS and evidence for hardness}
\label{Sec:Loops}

The discussion thus far in this work and in the literature far has focused on the hardness of GBS with unitary transformations drawn randomly from the Haar measure.
This requires implementing arbitrary unitary transformations, an onerous requirement experimentally. 
In fact, Ref.~\cite{Zhong2020} did not meet this requirement of being able to implement arbitrary unitary transformations as a result of the interferometer being a fixed non-programmable device. 
Furthermore, there is reason to believe that in the absence of error-correction methods for linear optics, scaling arbitrary programmable interferometers to large numbers of modes is infeasible.
This is because implementing an arbitrary unitary transformation requires decomposing it into beam-splitters and phase shifters and, assuming they are all applied locally, this leads to a deep optical circuit, whose depth linearly scales with its size.
Since photon loss scales exponentially with the circuit depth, these models necessarily become efficiently simulatable classically for sufficiently large numbers of modes~\cite{Qi2020,Oszmaniec2017,Garcia-Patron2019}.

On the other hand, naively reducing the depth without giving up gate locality is not an option for QCA either. 
This is because shallow one-dimensional (1D) circuits comprising local interactions with logarithmically scaling depths can be efficiently simulated classically as these do not generate enough long-range entanglement~\cite{Garcia-Patron2019,Qi2020,qi2020efficient}.

These results motivate a demonstration of QCA on random optical circuits with shallow depth but with gates that are long-range in 1D, for example on circuits with local interaction in higher than one dimensions.
In such a setting, a potentially reduced amount of complexity due to the reduced depth would be compensated by the large long-range entanglement generation thanks to the inclusion of long-range interactions.
Therefore, such an architecture would suffer less noise build up but still remain intractable for classical computers.
Indeed, models with shallow-depth but with long-range (in 1D) interactions provide a natural approach to demonstrating QCA in qubit systems~\cite{Bermejo-Vega2018}.

We address the challenge of the low-loss versus depth tradeoff by introducing high-dimension GBS, where programmable non-local gates are exploited to generate entanglement between distant modes. 
We show how high-dimensional GBS can be implemented scalably using optical delay lines. 
Before presenting the new architecture, let us recall the relevant notation on GBS and discuss its physical implementation.

\subsection*{A programmable architecture for high-dimensional GBS}

\begin{figure}
    \centering
    \includegraphics[width=\columnwidth]{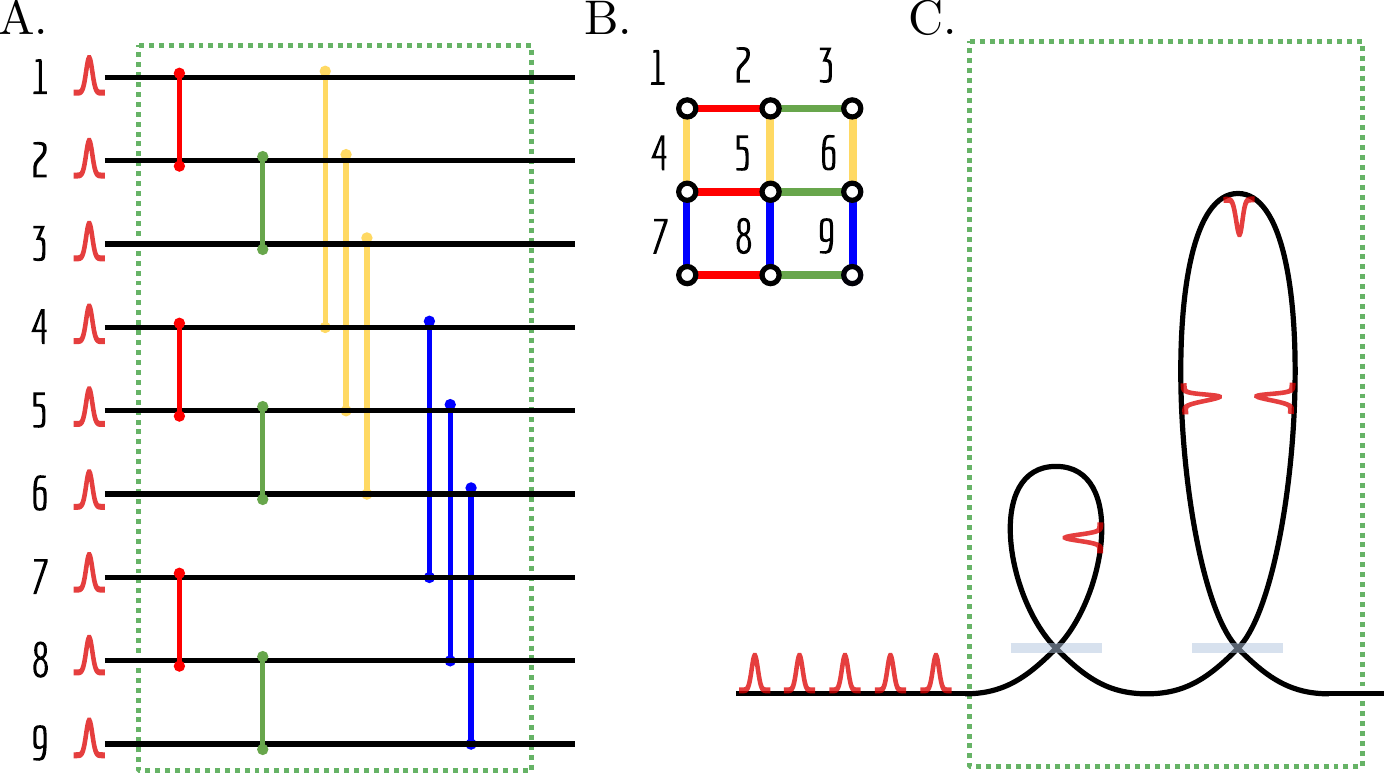}
    \caption{\raggedright \textbf{Different representations of a $D=2$-dimensional optical delay GBS instance
    with lattice size $a=3$.} (A)  Circuit representation. The vertical lines with dots at the end represent beam-splitters.
    (B) Bi-dimensional lattice representation. The vertices of the lattice represent the modes while edges represent beam-splitters. (C) Optical circuit representation. The modes are defined by time-bins traveling in a wave-guide. The horizontal gray slabs in the bottom of the delays represent the beam-splitters. 
    The number of cycles $C$ in a high-dimensional GBS instance correspond to applying multiple times the gates contained in the green-dotted box in Fig.~(A). This action physically maps to employing concatenating $C$ copies of the delays encircled in the green box in (C).
    Note that for simplicity we have not shown the photon-number detectors used to probe the quantum state at the end of the circuit.
    }
    \label{fig:2dgbs}
\end{figure}

Now we are ready to introduce high-dimensional GBS: a sampling task that retains the programmability of the photonic device, can reduce decoherence to a level that prevents classical simulability, and in which large amounts of multi-partite entanglement can be generated.
The last two requirements are to some extent at odds with each other: specifically, achieving long ranged interactions in fixed linear one-dimensional geometries  requires finding intermediary quantum systems to mediate interactions between far separated regions, which can lead to information leaking into the environment and require more challenging experimental conditions than all-optical experiments.
A way around this challenge is to consider two- or higher-dimensional geometries where quantum systems can interact with each in more than one direction. 
While the Google QCA experiment~\cite{arute2019quantum} involved interactions in 2D, our proposal can leverage photonics to implement distant non-local interactions, which can be equivalently considered as interactions in two or even higher than two dimensions.

More specifically, we show how the idea of using local interactions in high-dimensional spaces to generate large amounts of multi-partite entanglement can be naturally imported into photonic quantum computing by using optical delay lines and fast, programmable optical switches.
Before formally stating the problem of high-dimensional GBS we provide intuition for how to construct high dimensional lattices using minimal optical resources.
For the sake of concreteness and ease of visualization, we consider the generation of a lattice of size $a=3$ in $D=2$ dimensions where the vertices represent modes and the edges represent two-body gates.
A quantum circuit to achieve this connectivity and a representation of the obtained lattice are shown in Figs.~1(A) and (B), respectively.
Note that when the bosonic modes are represented as wires in a usual quantum circuit diagram the gates needed to prepare the state are highly non-local.
This is because circuit diagrams provide a representation where the modes are arranged linearly (in this case in the vertical direction of the page).
To show how optical delays provide a natural way to program short and long ranged interactions, consider first our temporal modes (pulses) prepared in squeezed-vacuum states arranged one after the other traveling along a single spatial mode, as schematically shown in Fig.~1(C).

We first consider how to achieve nearest-neighbor interactions using a delay line whose length equals the separation between the pulses.
As mode $i$ is in the delay line about to exit it, it will interfere with mode $i+1$ that is about to enter the delay line.
The beam-splitter mediating the interaction between these two-modes can be programmed allowing us to effect two-mode gates between nearest neighbors.
Such programmable and fast (i.e., with less than $50\,\text{ns}$ spacing) beam-splitters have been demonstrated using electro-optic modulation and have been used in the application of photonic quantum walks in the time domain~\cite{nitsche2018probing,schreiber2012a-2d-quantum}.

Now consider the second delay line, whose length is $a = 3$ times the separation between the pulses.
In this case, as mode $i$ is getting ready to exit the delay line, it will interfere with $i+a$ in the beam-splitter gate keeping the delay line.
This configuration allows interactions with range $a=3$ between the modes in the quantum circuit diagram in Fig.~1.
Note that this construction generalizes in a natural way to $D$ dimensions.
In particular, nearest-neighbor interactions in a $D$-dimensional space with $a$ lattice points per dimension (corresponding to gates with range $a^{D-1}$ in a circuit diagram) can be implemented using a circuit with $D$ optical delay lines implementing delays by amounts $\{1,a,a^2,\ldots,a^{D-1}\}$.
If the light is made to pass through $D$ such multiple such delay lines, with $C$ passes, then the effective transformation is composed of $C$ \textit{cycles} of local interactions in a $D$-dimensional lattice or equivalently, $C$ cycles of up to $a^{D-1}$-range gates in a circuit diagram. 

Having provided a quantum optical implementation of high-dimensional GBS we are now ready to formalize it by specifying four quantities: the squeezing parameter $r$, the lattice dimension $D$, the lattice size $a$, and the number of cycles $C$. An $(r,a,D,C)$-high dimensional GBS instance is constructed as follows:
\begin{enumerate}
  \item Prepare $M:=a^D$ single-mode squeezed vacua $|{r}\rangle^{\otimes M}$.
  \item For $\tau=1$, apply a beam-splitter $V$ to mode $i$ and $i+\tau$, where $i\in[0,M-\tau)$.
  \item Repeat Step 2 for $\tau=a^d$ for $d=0,\ldots, D-1$.
  \item Repeat Step 2 and step 3  a total of $C$ times.
\end{enumerate}
Having a physical architecture to implement high-dimensional GBS, we can now write down a loss budget to account for the bulk of the decoherence affecting our system.
Assume that the photon-number detectors used to probe our quantum state are limited by a rate of $\nu$ detections per second, for example as a result of the detectors dead times.
From this time scale we deduce a length scale $\ell = v/\nu$, where $v$ is the speed of light in the delay lines.
We associate with the length scale an energy transmission constant $\eta_{\text{unit-length}}$, which is simply the total energy transmission resulting from a propagation over a total length of $\ell$. 

Let us first study the case $C=1$. 
In this case, every mode will traverse $D$ beam-splitters (to access the $D$ different delay lines) and will propagate a total distance of $ \ell \times \sum_{i=0}^{D-1} a^i = \ell \times \dfrac{a^D-1}{a-1} \approx \ell a^{D-1} $ if $a \gg 1$. 
We can approximate the total transmission to scale roughly as 
\[
\eta = \eta_{\text{BS}}^D \eta_{\text{unit-length}}^{a^{D-1}} = \eta_{\text{BS}}^{D} \eta_{\text{unit-length}}^{M^{1-1/D}},
\]
where $\eta_{\text{BS}}$ is the beam-splitter transmissivity for programmable beam-splitters based on electro-optic modulation.
Note that in this case the loss scales sub-exponentially with the total number of modes.
To allow two or more circulations, one can consider $C\geq 2$ copies of the original $D$ delay lines, giving now an updated loss budget in which the modes traverse a length proportional to $C a^{D-1} = C M^{1-1/D}$ and will pass through $C D$ beam-splitters, still leading to sub-exponential loss accumulation.
An alternative to these $C$ copies of the delay lines is to consider a re-circulation loop similar to that proposed in Ref.~\cite{motes2014scalable}, which reroutes the output of the last delay line into the input of the first one.
The delay line used to implement the recirculation loop holds any modes that are not interfering inside the delay lines.
If the recirculator has a loss per unit length $\eta_{\text{unit-recirc}}$, the net loss scales as $\eta_{\text{unit-recirc}}^L$ where $L = a^{D} - \sum_{i=1}^{D-1} a^D = \Theta(M)$. 
Thus, depending on the exact setting, for a fixed $C$, the losses scale either exponentially (using recirculators) or sub-exponentially (considering multiple copies of the $D$ loops) with the number of modes.

We note that with current fiber-optic and photon number resolving technology, $\eta_{\text{unit-length}}$ can be as high as $0.998$;  $\eta_{\text{BS}}$ values of $0.9$ are expected or are observed in state of the art experiments such as Ref.~\cite{nitsche2018probing}.
With these, the transmission of an interferometer with parameters $(a=15, D=2, C=2)$ can be above $0.70$, and above $0.74$ for $(a=6, D=3, C=1)$.
These values promise an order or magnitude or more enhancements in loss values as compared to those expected in fully programmable GBS devices~\cite{taballione2020a-12-mode}.
As noted in Ref.~\cite{motes2015implementing}, interferometers implemented using loops will typically have unbalanced losses. The numbers quoted above assume the lossiest interferometer implementable in a loop based system, which is precisely the one in which each and every mode is fully transmitted into each loop.

From the formal description of high-dimensional GBS, the covariance matrix of the generated Gaussian state can be calculated in the usual manner.
In particular, we only need to specify the unitary matrix describing steps (2)-(4) above.
This unitary matrix is given by
\[\label{eq:defU}
U  = \bigotimes_{c=1}^C \bigotimes_{d=0}^{D-1} \bigotimes_{i=0}^{M-a^d} B_{i,i+a^d}(V)
\]
where $B_{i,j}(V)$ is an $M \times M$ unitary matrix that acts like the locally Haar-random beam-splitter $V$ in the subspace of modes $i$ and $j$ and like the identity elsewhere.
We denote by $\mathcal{U}$ the ensemble of linear-optical unitaries applied this way.
In Fig.~2, we show heat-maps of the unitary matrices associated with two typical instances from the distribution $\mathcal{U}$ over high-dimensional GBS instances.
Note that the structure of circuits considered allows for light from the first mode to be observed in any of the later modes, which leads to a large light cone that is somewhat different from the efficiently simulable circuits considered in recent Ref.~\cite{oh2021classical}.
From the description of the unitary matrices and the squeezing parameters, we have that the complex-valued adjacency matrix (as defined above) of the Gaussian state is dense, full-rank and given by $A = \tanh(r) U U^T$. 

While implementing a time-domain reconfigurable loop architecture as described above is not a straightforward task, several groups have performed experiments with tens of modes interfering in time-domain multiplexed configurations. These include time-bin~\cite{he2017time} and temporal-to-spatial encoded~\cite{wang2019boson}   boson sampling experiments~\cite{he2017time} and controllable photonic random walk over multiple time-bins~\cite{schreiber2012a-2d-quantum,lorz2019photonic}. Moreover, recent experiments have shown that is possible to operate with very high phase-stability~\cite{larsen2021deterministic}, high quantum-efficiency photon-number detection~\cite{arrazola2021quantum} and very low loss reconfigurable interferometric elements~\cite{takeda2019demand}.

\begin{figure}
    \includegraphics[width=\columnwidth]{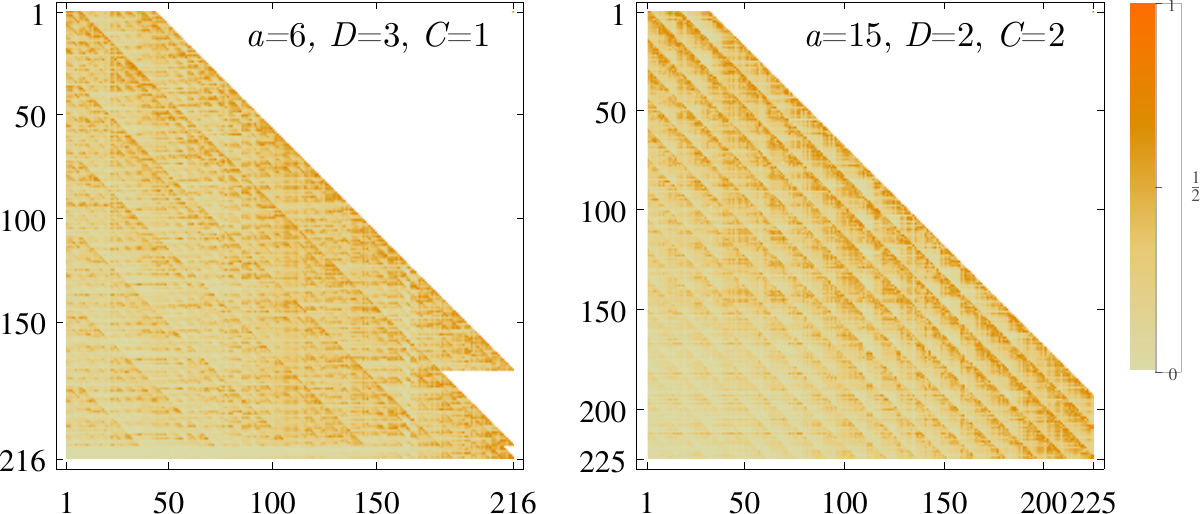}
    \caption{\raggedright%
    \textbf{Absolute values of the entries of the unitary matrices associated with two high-dimensional GBS instances drawn from $\mathcal{U}$.} On the left we show an $(a=6, D=3, C=1)$ instance and on the right we show an $(a=15, D=2, C=2)$ instance. 
    Note that we explicitly color the zero entries of the unitary white; thus the color scale is discontinuous at this end.
    \label{fig:Unitaries}}
\end{figure}

Finally, for the purpose of calculating outcome probabilities, squeezed states can be considered in the Fock basis as qudits that are entangled by the beam-splitter operations. 
This process, as with any other quantum circuit, can be represented as networks of tensors~\cite{lubasch2018tensor,dhand2018proposal}. 
In more detail, here the qudits are initially single-index tensors (vectors) that are contracted with four-index tensors representing the beam-splitters to build an open tensor network (TN), which can then be contracted to obtain the tensor of the final state.
The TN representing the state can be used to calculate probability amplitudes of measurements when the output indices of the TN are contracted with vectors representing measurement outcomes.
Similar TN-based techniques have been successful at delineating the QCA frontier in the context of random circuit sampling, and together with Hafnian based methods these will serve a similar purpose for high dimensional GBS.

\subsection*{Hardness for computing noisy probabilities in high-dimensional GBS} \label{Sec:restricted}

Here, we now argue for the hardness of computing output probabilities for the noisy, high-dimensional GBS setup. 
In particular, we show that hardness is present even in shallow depth noisy high-dimensional GBS architectures. 
This is in contrast to the results discussed earlier, where no restriction is made on the depth.

To do this, we simply observe that the previous argument for worst-case hardness, which depends on the noise being local and error-detectable, continues to hold for the limited-depth setup \cite{Brod2015}.
For average-case hardness of computing noisy probabilities, we again use a worst-to-average-case reduction.
However, the polynomial interpolation in this case is different, since a random instance is not Haar distributed any more but rather according to $\mathcal{U}$, the distribution over random instances of high-dimensional GBS.
To explain further, consider the usual interpolation $X(t) = (1-t) X + tY$, where $X(0) =X$ is drawn from $\mathcal{U}$ and $X(1)=Y$ is the matrix corresponding to a worst-case high-dimensional GBS instance.
In this case, there is no guarantee that the interpolated matrices $X(t)$ also correspond to high-dimensional GBS instances of small depth.
We get around this issue by choosing a gate-wise interpolation that is similar to that seen in RCS \cite{Bouland2019,Movassagh2019}.

We first define the problem of computing output probabilities of a restricted-depth high-dimensional GBS architecture.
\begin{problem}{$(\epsilon,\eta)$-HighDimensional-NoisyGBS-Probability}
Input & A noisy GBS instance drawn from $\mathcal{U}$ that can be implemented in $D$ dimensions with a constant number of cycles $C=O(1)$ with noise parameter $\eta$, and a description of a collision-free outcome $\mathbf{n}$ with $N=\mathsf{poly}(M)$ photons. \\
Output & With probability $\delta$ over instances, an estimate of $\Pr(\mathbf{n})$ to additive error $\epsilon$. \\
& With probability $1-\delta$, an arbitrary output.
\end{problem}

Similar to the previous results, we can again obtain an average-case hardness result that we state here and prove in the Supplementary Material.
\begin{thm}\label{lem_avgcasehardnoisy_HDGBS}
There exists a noise threshold $\eta_*$ and a sufficiently large polynomial such that the problem \textsc{$(\epsilon,\eta)$-HighDimensional-NoisyGBS-Probability} is $\#\mathsf{P}$-hard under $\mathsf{PH}$ reductions for any constant $\delta > 3/4$, $\eta \leq \eta_*$, and $\epsilon \leq 2^{-\mathsf{poly}}$.
\end{thm}

\subsection*{QCA frontier for high-dimensional GBS}
\label{Sec:Frontier}

The evidence presented above for the hardness of high-dimensional GBS comes from complexity-theoretic arguments, which are asymptotic in nature, i.e., they only specify how the hardness of a certain computation scales as the problem size is increased. 
For a finite sized device, we now address a complementary but more immediate question: how much actual computational power would a classical adversary need in order to generate samples similar to those from finite-sized noisy GBS devices?

This question can be addressed with different assumptions about the classical adversary.
The experiment can be benchmarked either against simulations that try to match a reasonable model of the experiment (constrained adversary) or against simulations that merely try to spoof a given test (unconstrained adversary). 
The latter approach would be more rigorous as it requires making fewer assumptions; but coming up with good spoofing methods is a problem beyond the scope of this work and should be seen as an ongoing community effort~\cite{renema2020marginal}.
Similar to the approach of the Google and USTC supremacy experiments~\cite{villalonga2020establishing,li2020benchmarking}, we focus on the former approach---with a classical adversary producing samples according to a noisy model distribution---because these samples are likely to perform at least as well as the actual device in suitable verification tests~\cite{kliesch2021theory}.
In other words, we assume a specific model of the imperfect GBS device, and we demand that the classical adversary generate samples that have a probability distribution that is sufficiently close in total variation distance to the probability distribution of this model. 
We note however that the chosen model might not have been verified against the actual experiment as this sample-efficient noise-model verification of QCA experiments is a challenging problem, especially for boson sampling and GBS.

We perform this benchmarking by simulating high-dimensional GBS with state-of-the-art algorithms on the current best supercomputers.
In particular, we consider the fastest algorithms based on computing probability amplitudes via Hafnians and via tensor-network contractions.
The former, Hafnian-based, algorithms have been optimized for simulating GBS and are not restricted to high-dimensional GBS~\cite{quesada2020exact}.
The latter, tensor network algorithms are well-suited for high-dimensional qudit circuits with shallow depth~\cite{gray2020hyper}.
We note that Ref.~\cite{qi2020efficient} also provides a path to simulating lossy GBS if the losses scale exponentially with the system size, but these results are not applicable for high-dimensional GBS, where the losses can scale subexponentially.
By benchmarking against these algorithms we demonstrate that high-dimensional GBS experiments feasible with current optical technology are well beyond the reach of the biggest supercomputers.

\section*{Discussion}

In this work, we have proposed a new experimental architecture for Gaussian boson sampling and provided asymptotic evidence for the hardness of Gaussian boson sampling in this specific context, bridging the gap between theory and experiment.
We have also benchmarked today's best-known algorithms at simulating such an experiment, obtaining complementary evidence that a reasonably-sized setup would outperform classical supercomputers at this task.
Still, some theoretical questions are outstanding.
\begin{enumerate}

    \item We have been able to show that two plausible conjectures in random matrix theory allow us to obtain the hiding property for a noiseless GBS set up, without restrictions on the number of active modes. Can we obtain a similar hiding property for the high-dimensional GBS set-up introduced in this work? Is this also possible in the presence of noise? Answering these questions is crucial for extending the hardness of computing output probabilities to the hardness of approximate sampling from experimentally realizable distributions.
    \item Informally, the anti-concentration conjecture for boson sampling (or GBS) states that the output probability of a random instance is unlikely to be very small. If this conjecture was true, then now-standard arguments can show that the output probability corresponding to an approximate sampler is, with high probability, a good multiplicative estimate to the ideal output probability. Proving this conjecture true, in either the case of boson sampling or GBS, would give increased evidence to support the goal of proving QCA via photonics. 
    A proof of such a conjecture is challenged
    due to the fact that tools of unitary designs
    \cite{Hangleiter17AnticoncentrationQuantumSpeedup}
    are presumably unavailable in the bosonic setting~\cite{Aaronson2013}.
    \item 
    Notwithstanding, it would be insightful to compute the second moments $\mathbb E_{X \sim \mathcal{G}(0,1/M)} |\mathrm{Haf}(XX^T)|^4$ for the distribution we have found to characterize GBS problem instances. 
    These moments thus characterize the so-called \emph{collision probability} of seeing the same outcome twice in an experiment, which in turn can be related not only to anti-concentration but also the verifiability of approximate GBS from samples, 
    thus shedding some light on the structure of the GBS output distribution.
    \item 
    An important task in demonstrating QCA is to verify that the performed experiment indeed contains a non-trivial quantum signal that cannot be efficiently spoofed. 
    The Google QCA demonstration relied on linear cross-entropy benchmarking fidelity, and the USTC experiment used a heavy-output generation (HOG) ratio test as an alternative path to verifiable hardness.
    Whether the HOG-ratio test can be spoofed efficiently by a classical adversary such as the algorithms considered in Refs.~\cite{kalai2014gaussian,renema2020marginal} is an open problem.
    \item The recent result in Ref.~\cite{oh2021classical} presents a classical algorithm for the simulation of high-dimensional boson sampling experiments in certain regimes. 
    As described, this algorithm is not applicable to the architecture we propose in this work. 
    Extending the algorithm to be relevant to the present architecture is an open problem.
    \item With current optical technology, loss is the dominant source of noise in any GBS experiment. Consequently we were motivated to obtain hardness results for computing the output probabilities of a GBS experiment in the presence of significant photon loss. 
    It is natural to investigate if similar hardness results can be obtained in the presence of other possible sources of experimental noise such as such as mode mismatch, multiple Schmidt modes, interferometer phase drift and detector dark counts. 
    \item It is a challenge to the community, after all, to relate boson sampling closer to practically important computational tasks and to identify new applications.
\end{enumerate}

In summary, this work brings the demonstration of QCA on a programmable photonics device closer to reality. It addresses previously outstanding theoretical challenges in the field by providing stronger evidence for the hardness of GBS.
Crucially, we have presented a novel architecture for high-dimensional GBS using optical delay lines that promises low levels of noise without compromising on its programmability. 
We benchmarked this architecture against the best available classical simulation algorithms and found that already experiments involving a moderate number of modes are far beyond reach for those algorithms.

We close by briefly commenting on the experimental prospects of realizing high-dimensional GBS. Since high-dimensional GBS can be implemented in the time domain according to the scheme presented in Fig.~1, only a single squeezer and a single detector are required.
If multiple detectors are available, these can be de-multiplexed using optical switches in order to increase the effective repetition rate of the experiment and reduce the length of the delay lines.
Especially promising is the case of $D = 3$, $a = 6$, $C = 1$, which can be implemented with only three optical delay lines and three each of re-programmable beam-splitters and phase shifters.
Assuming reasonable values of squeezer out-coupling losses, free-space to fiber coupling loss and detector efficiency \cite{larsen2021deterministic,Larsen2019,arrazola2021quantum,takeda2019demand}, we estimate that such a setup can be built using current optical technology with around 40\% transmission, higher than that enabled by the ultra-low non-programmable loss interferometer in the USTC experiment. 
Such a setup would enable the largest demonstration of QCA yet with a mean detected photon number of 80 in a programmable device with 216 total modes. 
We hope that this work stimulates such developments.

\section*{Material and Methods}
\subsection*{Computational task: Sampling from lossy GBS with finite Fock cut-off}

Before looking into concrete strategies for the simulation of GBS we detail the computational task performed by the GBS device and discuss some differences between the task and our simulation.
The experimental device samples from a lossy GBS distribution with a finite Fock-basis cut-off, which results from detector limitations.
In order to identify a range of parameters where this task is hard to simulate classically, we benchmark it against classical simulations.
The simulations that we compare are somewhat different from the exact task performed by the experiment but in such a way that is advantageous to the classical simulations, thus providing stronger evidence for the large computational cost of high-dimensional GBS.
We now discuss these differences.

The first point of difference is the Fock cut-off, i.e., the number of Fock or photon-number levels considered in each mode.
Both Hafnian and tensor-network simulations are performed in Fock basis and their performance is thus sensitive to the Fock cut-off.
This cut-off must be chosen carefully because the squeezed state inputs in GBS have non-zero support on high Fock numbers (which could be infinite in the ideal case)~\cite{quesada2020exact}.
For Hafnian-based simulations, the Fock cut-off $c$ will lead to a constant prefactor $2^{c}$ ($2^{c/2}$) in the runtime for calculating mixed-state (pure state) probabilities that would appear in sampling methods. 
Similarly, for tensor network simulations, this cut-off sets the qudit dimension in the calculation, which is also the base of the exponential function describing the time and space cost of contracting the tensor network. 
Note that squeezed states of light require that we use local Hilbert spaces with at least dimension 3, since truncating a squeezed state to the first two levels of the Fock ladder will project it into the vacuum, since $\langle{1}|{r}\rangle = 0$.
Furthermore, using a Fock cut-off of $3$ in the beam-splitter gates leads to highly inaccurate simulations as the beam-splitter transformations on a limited Fock subspace no longer preserve photon numbers.
In other words, choosing higher Fock cut-offs will lead to more accurate but more expensive simulations.
Hence, we use a cut-off of 4 to give a conservative estimate on the computational cost, even though this cut-off would lead to inaccurate classical simulations.

\begin{figure}
\centering
\includegraphics[width=\columnwidth]{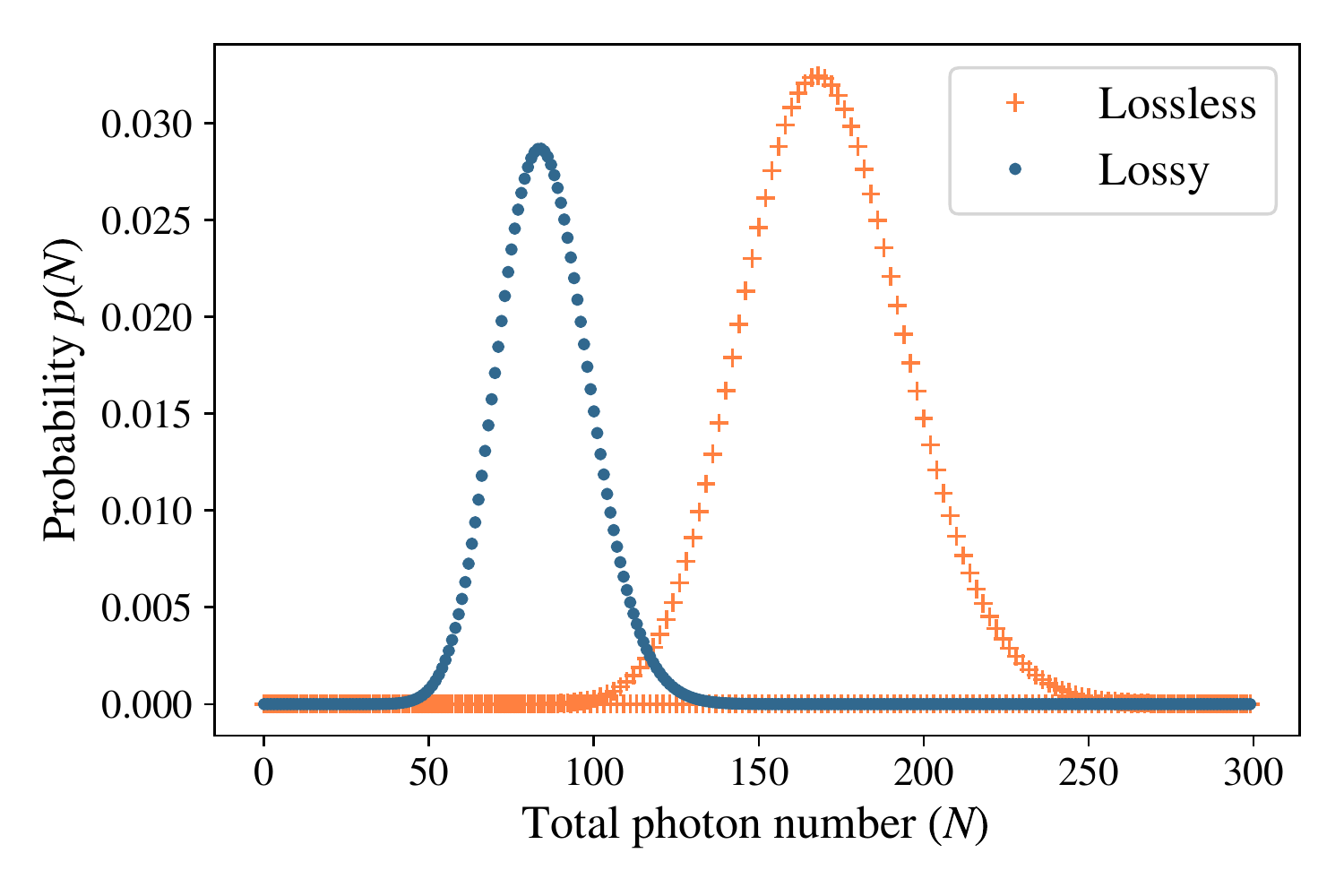}
\caption{\raggedright \label{fig:total_n} 
\textbf{Distribution of the total photon number for $M=216$ single mode squeezed states with squeezing parameter $r=0.8$.} We assume a total transmission of $\eta = 0.5$ (corresponding to roughly 3 dB of loss) for the lossy distribution. Note that the lossless distribution has no support on odd numbers of photons, which explains why visually it looks as if it has more area under the curve.}
\end{figure}

A second point of difference is that our simulations deal with the case of simulating pure states with photon numbers equal to the lossy distribution.
This is a reasonable simplification, since as shown in  Ref.~\cite{quesada2020quadratic},
simulating pure or mixed state GBS has the same complexity as calculating a number of pure-state probability amplitudes proportional to the number of modes in the system.

Before describing the effect of loss on the two simulation methods, we discuss the effect of loss on the number of detected photons.
In Fig.~3, we plot the lossless and lossy (transmission $\eta= 0.5$ $\approx 3$ dB loss) distribution for $M=216$ modes and squeezing parameter $r=0.8$. 
These parameters have been chosen to correspond with an $(r=0.8, a=6, D=3,C=1)$ high-dimensional GBS instance with experimentally reasonable loss budgets. The squeezing parameter $r=0.8$ is chosen to be within reach of current sources of single-Schmidt mode degenerate squeezed light~\cite{vernon2019scalable}. 
Note that the lossy distribution has smaller mean and variance than the lossless one \cite{dodonov1994photon}, indicating that it becomes easier to simulate a lossy distribution as the transmission $\eta$ is decreased. For example, the outcome with the highest probability in the lossless distribution 
\begin{equation}
n^* = 2\left  \lfloor{ \left( \tfrac{M}{2}-1\right)\sinh^2 r}\right \rfloor =168
\end{equation}
has a probability of $7.28 \times 10^{-8}$ under the lossy distribution. The leftwards shift of this distribution will in general be present whenever loss acts on a pure state. 
For $M$ identical squeezers (with squeezing parameter $r$) undergoing loss by energy transmission $\eta$, the mean and variance contract at least proportionally to $\eta$
\[
\mathbb{E}(n) =  \eta M \sinh^2 r, \\
\text{Var}(n) = \eta  M \sinh ^2 r (1+ \eta  [1+ 2 \sinh^2 r]),
\]
confirming our intuition, and moreover showing that the prevailing sources of decoherence in photonic sampling problems behave differently from the ones in random circuit sampling implemented in superconducting circuits, where noise makes the output probability distribution become uniform~\cite{Boixo2018}.

We now focus on the case of Hafnian-based algorithms.
The cost of calculating the relevant probabilities depends only on the number of photons detected.
Calculating a photon-number probability $\Pr(\bar{n})$ of a mixed state is roughly quadratically more expensive than calculating a pure state probability of an event with the same number of photons~\cite{quesada2019simulating}.
However, the cost of sampling pure and mixed states is similar.
This is because lossy GBS states are classical mixtures over a displacement parameter of pure Gaussian states.
Therefore, it is possible to sample from a lossy state by sampling from the convex hull parametrized by the displacement parameter and then sampling from the pure state.
Thus, sampling lossy GBS states has similar computational cost as sampling pure states with the same number of photons.

Likewise, for the tensor-networks based algorithms, the cost for mixed state calculations would scale at least quadratically worse as compared to pure state calculations. This is because twice as many tensors are involved in a mixed state calculation, analogous to the quadratic overhead of keeping track of the density matrix as compared to a pure state. Note that for noisy random circuit sampling  of qubits, one can trade fidelity for sampling speed \cite{markov2018quantum}. As opposed to GBS, this improvement is possible because in RCS, the amplitudes of the different Feynman-like paths that appear when slicing through two-body gates in the circuit are comparable.
Moreover, this improvement is useful as long as the Schmidt-rank of the two-body gates used to generate entanglement is small, which is not the case for the beam-splitter.
Furthermore, the state vectors associated with two different paths are approximately orthogonal. 

A final point of difference between our simulations and the actual experiment is that while our run-time estimates are for the \textit{calculation} of the GBS probabilities, an actual experiment \textit{samples} from this distribution. 
Despite this difference, our simulations allow a fair benchmarking of the quantum device because current state-of-the-art algorithms possess similar complexities of sampling and calculating probabilities. We moreover give the classical adversary an extra advantage in that we allow it to assume that only pure-state output probabilities need to be calculated for sampling as opposed to the quadratically slower mixed-state output probabilities, since as explained above, mixed Gaussian states are convex mixtures of pure ones.
For the case of Aaronson-Arkhipov Boson Sampling, this argument was shown to be correct by using \emph{Markov-Chain Monte Carlo (MCMC)} methods to generate samples from the ability of calculating pure-state output probabilities~\cite{neville2017classical}.

In summary, we provide maximal advantage to a classical adversary by choosing a low Fock cut-off, by performing pure state simulations with low photon numbers and by estimating time for computation rather than sampling (which is at most polynomially slower using currently known methods).
This advantage ensures that despite improvements in the classical algorithms, the space of parameters that are hard to simulate classically remain so.

\subsection*{Hafnian-based algorithms}
Consider now the probability amplitudes of $n$-photon events by evaluating the Hafnian. 
Similar benchmarkings have been performed in the past for the calculation of permanents~\cite{wu2018benchmark} (relevant to boson sampling) and Torontonians~\cite{li2020benchmarking} (relevant to GBS with threshold detectors).
For either of these two tasks, the time complexity of calculating a probability corresponding to an $n$-photon event scales like $O(\text{poly}(n)2^n)$, which is quadratically worse than for GBS, which scales as $O(\text{poly}(n)2^{n/2})$.
For the case of boson sampling, this difference stems from the fact that any probability amplitude with $n$ photons maps exactly to a GBS instance with $2n$ photons.
For the case of threshold detection it stems from the fact that one cannot assign probability amplitudes to a measurement that is not rank-one, like the POVM representing a ``click'' which is a coarse-graining of all the projectors with nonzero photons. 
In any case, for either of these tasks, benchmarks up to $n=50$ have been carried out requiring on the order of two hours for boson sampling using Tianhe-2~\cite{wu2018benchmark} and on the order of 20 hours for GBS with threshold detectors in Sunway TaihuLight~\cite{li2020benchmarking}. 

\begin{figure}\centering
\includegraphics[width=\columnwidth]{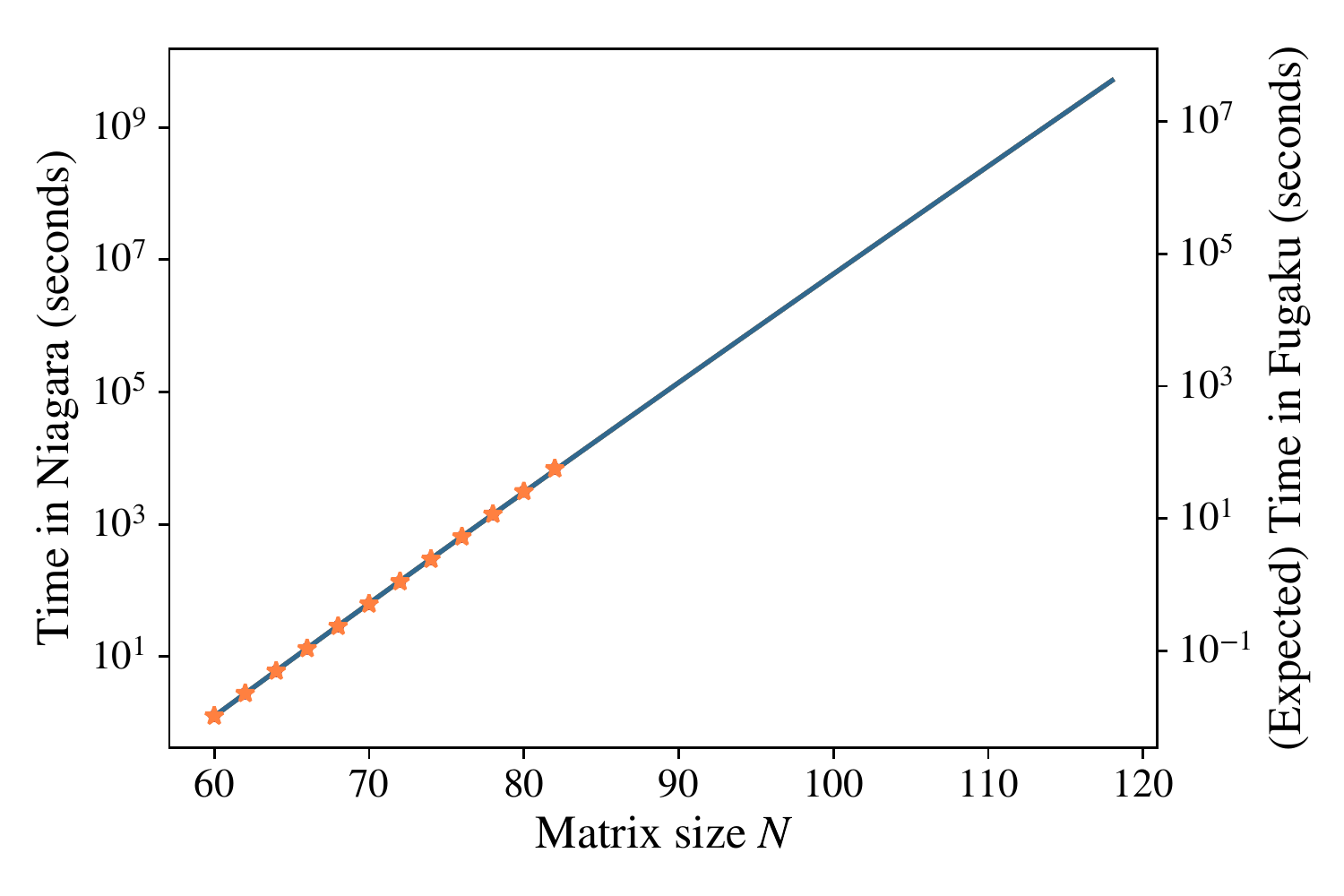}
\caption{\raggedright \label{fig:Hafnian_benchmarking}
\textbf{The time cost of calculating a Hafnian of size $n$ in double precision.} 
The stars indicate actual sizes computed in the Niagara supercomputer~\cite{ponce2019deploying}.
The blue line is a fit to $t_{\text{Niagara}}(n) = c_{\text{Niagara}} n^3 2^{n/2}$ with the only fitting parameter $c_{\text{Niagara}}= 5.42 \times 10^{-15} \ \text{s}$.
The standard deviation of fitting parameter $c_{\text{Niagara}}$ is  $1.2 \times 10^{-16} \ \text{s}$, which would give error bands thinner than the width of the line.
We find an equivalent expected time in Fugaku, among the most powerful supercomputers, by considering the ratio of their Rmax scores (maximal LINPACK performance achieved) giving their performance in number of floating point operations per second. 
The conversion factor between the left scale for Niagara and the right scale for Fugaku is the ratio of Rmax values of Fugaku and Niagara, or equivalently $c_{\text{Niagara}}/c_{\text{Fugaku}}=122.8$. 
Note that since the computation of Hafnians can be broken into the independent calculation of an exponential number of summands (known as an \emph{embarrassingly parallel} computation) this scaling is expected to be quite accurate.}
\end{figure}

If the matrix has no special property, like being low-rank, non-negative, banded, or sparse, the best known algorithms to calculate the Hafnian will scale like $O(n^3 2^{n/2})$ for a matrix of size $n \times n$.
The adjacency matrices generated in high-dimensional GBS do not have any of these properties.
In Fig.~4, we show the results of our benchmarking by implementing the Hafnian algorithm from Ref.~\cite{bjorklund2019faster} using a task-based approach implemented in Ref.~\cite{huang2020cpp}. 
Even for  shared-memory CPU architectures, our new task-based implementation achieves a speed up of about $5 \times$ with respect to the current OpenMP implementation described in Ref.~\cite{gupt2019walrus}.

Based on these benchmarks, we estimate that Fugaku, among the current most powerful supercomputers in the world, would require around 14 hours to compute the Hafnian of a $100 \times 100$ matrix. Thus, if the total-photon-number distribution of a given GBS setup has significant support past 100 photons, there will be a proportionally significant number of probability amplitudes that will require at least 14 hours in Fugaku to be computed.

We can get an estimate of the average time it would take to generate a sample by averaging the time it take to generate a sample with $n$ photons over the probability distribution of $n$ photons. Using the same averaging procedure,  but applied to clicks instead of photons and assuming an overhead of 100 between computing probabilities and generating samples, the authors of Ref.~\cite{Zhong2020} estimate that Fugaku would require around $1.9\times 10^{16}$ seconds to generate roughly the number of samples that their experiment produces in 200 seconds at MHz clock speeds.

For the lossy instance considered in Fig.~3, we find that on average Fugaku would require $F_{\text{amplitudes/samples}}\sum_{n=0}^{n_{\max}} p_{\text{lossy}}(n) c_{\text{Fugaku}} n^3 2^{n/2} \approx 4\times 10^{7}$ seconds to generate one sample. 
In this estimate we do not extend the sum to all possible photon numbers but only up to those that have a chance of more than $10^{-7}$ to occur, which happens at $n_{\max} = 166$ and moreover assume a reasonable overhead of $F_{\text{amplitudes/samples}} = 100$ for the calculation of probability amplitudes vs. samples. As noted earlier, the complexity of generating a sample for a mixed or pure Gaussian state is proportional to that of calculating a probability amplitude~\cite{quesada2020quadratic} and the number of modes (in our case 216), thus, using a factor of 100 is likely an underestimate.

In order to match the number of samples generated in seconds in a quantum device operating at 10 KHz would require $6.8\times 10^{15}$ seconds.
Thus, the computational cost of an $(r=0.8, a=6,D=3)$-high-dimensional GBS instance, with $3$ dB of loss is on par with the expected classical complexity of the USTC experiment with the added advantage of being programmable and much closer to the collision free-regime: the expected classical complexity of an experiment like the one just described is similar to the expected time complexity of the USTC experiment~\cite{Zhong2020}. However, besides the obvious disadvantage of programmability, their experiment is much farther away from the collision-free regime in which computational complexity theoretic results guarantee the intractability of GBS.

For example, if the USTC experiment had been performed with PNR detectors we would find that their photon number distribution has mean and standard distribution $83.3 \pm 20.1$ over 100 modes (where we assume the squeezing parameters quoted in Ref.~\cite{Zhong2020} and a net transmission of $\eta =0.3$). Note that even within the first standard deviation one is already beyond the total number of modes. This should be contrasted with a distribution like the one in Fig.~3, for which we find $85.2 \pm 13.9$ over $216$ modes.

\subsection*{Tensor networks methods}

Another promising method to calculate the probability amplitude of high-dimensional GBS is using tensor-network contractions.
This has been the strategy of choice for classical adversaries to superconducting circuits performing random circuit sampling~\cite{villalonga2020establishing,markov2018quantum}.
For an overview of tensor network algorithms to simulate quantum circuits, see Ref.~\cite{Biamonte2017}.

In this section, we find that tensor network algorithms can simulate two-dimensional lossy GBS experiments on ~200 modes in a reasonable amount of time.
This motivates going to a higher dimension, $D=3$.
We find that after making several allowances to the classical algorithm and accounting for tremendous improvements in classical hardware, one of the fastest supercomputers in the world, Fugaku, would take $\sim 10^{20}$ seconds to simulate a 3-D experiment on 216 modes running for 200 seconds.

Any given quantum circuit can be written as a network of tensors such that each input quantum state is a rank-1 tensor, each gate acting on $\ell$ components is a rank-$2\ell$ tensor, and each measurement operator is a rank-1 tensor~\cite{gray2020hyper}.
The probability amplitude for the quantum circuit can then be calculated by contracting the tensor network, i.e., by summing over all the indices of the tensor network.
However, there are multiple different orderings (paths) in which the different indices of a tensor network can be contracted, which influence the contraction runtime.
In fact, for general instances, the problem of finding optimal contraction paths for minimizing the time required to compute amplitudes has been shown to be $\mathsf{NP}$-hard ~\cite{pfeifer2014faster}, while actually performing the contraction is $\#\mathsf{P}$-hard~\cite{damm2002complexity,gray2020hyper}.
For some of the first classical benchmarking proposals of random circuits, the contraction paths were hand-picked by the researchers~\cite{villalonga2019flexible}. 
More recently, excellent randomized algorithms have been introduced to find contraction paths that have been shown to improve on previous results~\cite{gray2020hyper}.

A second important practical consideration for tensor-network contraction is that there is a trade-off between space and time complexity. That is, one can speed up significantly the contraction of a tensor network at the expense of assuming access to large amounts of memory. 
A systematic way to reduce the memory footprint of a tensor network contraction (at the expense of decreasing the speed of the computation) is to use a technique known as slicing, also known as variable projection or bond cutting~\cite{villalonga2019flexible}.

Unlike for Hafnian methods where one does not need to specify much of the structure of the circuit, this information is vital in understanding the performance and limitations of tensor network simulations. 
As before, we fix the squeezing parameter $r = 0.8$ and assume net end-to-end transmission of $\eta = 0.5$. 
With these parameters and first assuming $D=2$, we need at least $a=14$ lattices sites per dimension to get to a mean photon number at the detectors (i.e. after loss) of $\mathbb{E}(n)\sim 80$.
For a single cycle $C=1$ we use a tensor-network contraction algorithm called \texttt{cotengra}~\cite{gray2020hyper} together with Fugaku's LINPACK benchmark to find that this supercomputer would require less than 100 microseconds to contract the tensor network.
Thus, for 2-dimensional instances up to this size it is necessary to consider more than one cycle, implying the construction of either $D$ extra delay lines or adding a circulator, both of which will adversely affect the net transmission.

\begin{table}
  \begin{center}
    \begin{tabular}{l|l|l} 
      Number of & Expected  Time& Size of the \\
      lattice points $(a)$ & in Fugaku (seconds) & largest tensor \\
      \hline
      4 & $1.65 \times 10^{-1}$ & $4.39 \times 10^{12}$\\
      5 & $4.56 \times 10^{5}$ & $4.61 \times 10^{18}$\\
      6 & $2.11 \times 10^{14}$ &$7.92 \times 10^{28}$\\
    \end{tabular}
  \end{center}
\caption{\raggedright \label{tab:TN} Benchmarks for a $D=3$ high-dimensional GBS instance with minimal Fock space cut-off $c=4$.  The first column gives the number of lattice points, from which the number of modes follows $M = a^3$. The second column is the expected run time in Fugaku. This time is obtained by estimating the number of floating point operations required to contract the tensor using \texttt{cotengra}~\cite{gray2020hyper} and converting this into a time by using the Rmax floating point operation per second score for Fugaku. Note that \texttt{cotengra} implements randomized algorithms, thus for each problem size we run it 200 times and confirm that after the first 100 runs there is no significant variation in the best score found.
The last column gives the number of elements of the largest tensor ever needed to be stored in memory during the contraction. Note that this places restrictions on the RAM available in each of the nodes of a supercomputer. In particular the nodes in Fugaku have up to 32 Gb of RAM allowing to store on the order of $4\times 10^{9}$ 64 bit floating point numbers, thus an $a=6$ instance will  far exceed the required capacity of a single node requiring distributed storage and thus subsequent hit in efficiency due to communication complexity.}

\end{table}

This motivates considering the next dimension, $D=3$. 
For this case, and fixing the number of cycles to $C=1$, we find that we need at least $a=6$ to have a mean photon number on the order of 80 at the detectors, which would provide a non-trivial support on photon numbers that are beyond the reach of the Hafnian algorithms described above. 
In Table~1 we show the time it would take Fugaku to contract different three-dimensional GBS circuits for different lattice sizes. 

Note that even allowing for a hypothetical scenario in which the RAM of each of its nodes has been expanded by about 19 orders of magnitude, it would take Fugaku on the order $2.11 \times 10^{14}$ seconds to calculate a contraction with a minimal (and highly inaccurate) cut-off of $4$.
In reality, it is infeasible to fit the computation in the memory or even the hard disks of individual nodes, so slicing would be required, which can lead to astronomical overheads over this idealized estimate.
Even without this overhead and assuming that generating a sample is as expensive as calculating a probability, simulating a 200 second 10kHz experiment would require over $4\times10^{20}$ seconds. 
Of course, we remind the reader once more that a direct calculation of output probabilities is not what the experiment does but only what one model of the experiment, and there may be more efficient methods for simulating a verifiable experiment.

Based on the evidence presented above, a high-dimensional GBS instance with squeezing parameter $r=0.8$, in $D=3$ dimensions, with $a=6$ modes per dimension or a total of $216$ modes and a single cycle $C=1$ is well beyond the capabilities of current simulation methods based either on Hafnian calculations or tensor network contractions, even when losses of around 3 dB ($\eta \sim 0.5$) are present. This significant computational gap is present even after the fact that we allow the classical computer to ignore significant overheads in terms of cut-off, number of modes and samples-to-amplitudes conversion.
These experimental parameters we propose are within the reach of current photonics technology and their implementation using time-domain multiplexing can be achieved with a significantly reduced number of components. 

Note.--- After this submission, we became aware of a recent work \cite{Zhong2021} on an upgraded version of the experiment done in Ref.~\cite{Zhong2020}.

\acknowledgements{
We thank Juan Miguel Arrazola, Luke G. Helt, Matthew Collins, Fabian Laudenbach, Ilan Tzitrin, and Zachary Vernon for helpful discussions.
We also thank the authors of Ref.~\cite{bouland2021noise} for sharing an early version of their manuscript.
M.~H., M.~I.\ and D.~H.\ thank Karol Zyczkowski for enlightening discussions regarding the distribution of COE sub-matrices. 
\textbf{Funding:} A.~M.\ is funded by Mitacs Accelerate Program.
M.~H., M.~I.\, J.~E.\, and D.~H.\ are funded by the DFG (EI 519/21-1, EI 519/9-1, EI 519/14-1, CRC 183), the MATH+ Cluster of Excellence, the BMBF (HYBRID), the Einstein Research Foundation (Einstein Research Unit on near-term quantum devices), BMBF (QPIC), BMBF (PhoQuant), and the European Union’s Horizon 2020 research and innovation programme under grant agreement No. 817482 (PASQuanS).
The authors thank SOSCIP and SciNet for their computational resources. Computations have been performed on the Niagara and the Mist supercomputers at the SciNet-SOSCIP HPC Consortium. SciNet is funded by: The Canada Foundation for Innovation, the Government of Ontario; Ontario Research Fund - Research Excellence, and the University of Toronto. SOSCIP is funded by the Federal Economic Development Agency of Southern Ontario, the Province of Ontario, IBM Canada Ltd., Ontario Centres of Excellence, Mitacs and Ontario academic member institutions.
\textbf{Author Contributions:} Theory work was completed by authors A.~D., A.~M., M.~H., M.~I., H.~Q., J.~E., D.~H., and B.~F. Authors T.~V., N.~Q., L.~M., J.~L.~completed the experimental design and benchmarking work. Author I.~D.~managed the project and made contributions to both the theory and benchmarking work. All authors discussed the results and contributed to writing the manuscript.
\textbf{Competing Interests:} B.~F.~and A.~D.~acted as paid consultants for Xanadu Quantum Technologies while parts of this work were performed. The authors declare no other competing interests.
\textbf{Data and Materials Availability:} All data needed to evaluate the conclusions in the paper are present in the paper and/or the Supplementary Materials.
}\\

\onecolumngrid
\renewcommand{\theequation}{S\arabic{equation}}
\renewcommand{\thefigure}{S\arabic{figure}}
\setcounter{equation}{0}
\setcounter{figure}{0}
\appendix

\section*{Evidence for hiding in GBS}
 \label{sec_apdx_hiding}

In this section, we characterize the distribution of the symmetric product of $N\times K$ sub-matrices of $M\times M$ Haar-random unitaries. As described earlier, the Hafnian of such symmetric products determines the output distribution of GBS. 
Here, we give evidence that this distribution tends to the distribution of the symmetric product $XX^T$ for $X$ being an $N \times K$ Gaussian matrix. 
In GBS, this ensures the hiding property since a small $N\times N$ symmetric Gaussian matrix
$XX^T$ can be hidden in a large symmetric unitary matrix $U I_K U^T$ for any
$K\geq N$. 
Since any particular sub-matrix cannot be distinguished from any other such sub-matrix of
the same size, this enforces the constant error budget of an adversarial sampler to be
roughly equally distributed across all outcomes.

In particular, we consider three regimes---with respect to the relations between the total number of photons at the output ($N$), the number of input squeezers ($K$), and the number of modes ($M$)---in order to provide evidence for Conjecture 1.
This conjecture relates the following ensembles of random matrices.
\begin{enumerate}
    \item $\mathcal{H}^M_{N,K}$: The ensemble of $N \times K$ sub-matrices of Haar-random unitaries $U \in U(M)$.
    \item  $\mathcal G_{N,K}(\mu, \sigma^2)$: 
    The ensemble of $ N \times K$ matrices with independent and identically distributed 
    (i.i.d.) complex normal entries with 
    mean $\mu$ and variance $\sigma^2$.
    \item  COE$^M_{N,K}$: The ensemble of matrices $VV^T$ where $V \sim \mathcal{H}^M_{N,K}$.
    \item $\mathcal G^{\text{sym}}_{N,K}(\mu,\sigma^2)$: The ensemble of matrices $XX^T$ where $X \sim \mathcal G_{N,K}(\mu,\sigma^2)$.
\end{enumerate}
As the conjecture might be interesting for random matrix theory in itself, we will abstract away the meaning of the parameters $K, N, M$.
\begin{conj}[Hiding in GBS]
\label{conj:hiding_gbs}
For any $K$ such that $N \leq K \leq M$ the following statements are true:
\begin{enumerate}
  \item \label{it:asymptotics} For $ M \in O(N^{2+\epsilon})$ and $\epsilon \in (0,1]$, COE$^M_{N,K}$ asymptotically approaches $\mathcal G^{\text{sym}}_{N,K}(0,1/M)$ in probability in terms of the entrywise max-norm. 

\item \label{it:tvdist} There exists a polynomial $p$ such that for any $\delta > 0$ and $M \geq p(N)/\delta$, the total-variation distance $\|{\cdot\|}_{TV}$ between COE$^M_{N,K}$ and $\mathcal G^{\text{sym}}_{N,K}(0,1/M)$satisfies
  \begin{equation}
    {\|\text{COE}^M_{N,K} - \mathcal G^{\text{sym}}_{N,K}(0,1/M)\|}_{\text{TV}} \in O(\delta). 
  \end{equation}
\end{enumerate}
\end{conj}

Here, we give analytical evidence that the characterization of Conjecture \ref{conj:hiding_gbs} holds true in the extreme cases of $K=N$ and $K=M$ for $M$ growing fast enough with $N$ and numerically show that it is true for any $K$ such that $N\leq K\leq M$. 

In the first regime we consider $K$ is such that $M \in \Omega(K^5 \log^2 K)$ and $N = K$. 
This regime closely resembles the one in the original boson sampling proposal (thus we refer to it as the ``AA regime'') for which we will see that both parts 1.~and 2.~are provably true. 
In this regime, Aaronson and Arkhipov~\cite{Aaronson2013} have proven that all $N\times K$ sub-matrices of Haar-random linear-optical unitaries $U$, are approximately Gaussian distributed. In particular they show that $\mathcal{H}^M_{N,K}$ asymptotically approaches $\mathcal{G}_{N,K}(0,1/M)$ as well as bounding the rate of convergence by showing  that the total-variation distance satisfies
\begin{equation}
{\|\mathcal{H}^M_{N,K} - \mathcal G_{N,K}(0,1/M)\|}_{\text{TV}} \in O(\delta)
\end{equation}
for $M \geq (N^5/\delta) \log^2 (N/\delta)$ \cite{Aaronson2013}. Using this we can directly see that Conjecture 6 is also true in the ``AA regime''. 

On the other end of the spectrum, we consider the regime in which $K=M$ where part 1.~of the conjecture is provably true. For this case, Jiang~\cite{jiang_entries_2009} has shown that the distribution of $N \times N$ sub-matrices of $M \times M$ COE matrices for $M \in o(\sqrt N/ \log N)$ asymptotically approaches the distribution of matrices $XX^T$, where $X\sim \mathcal  G_{N,M}(0,1/M)$.

Finally, there is the intermediate regime in which $M^{1/5} \lesssim K < M$.  
This regime interpolates between the two extreme regimes of very small, square sub-matrices of $U$ and very short, fat sub-matrices of $U$.
A priori, there is no reason to believe why the behaviour should differ from the extreme regimes. 
Indeed, for this regime we can provide numerical evidence for both parts of Conjecture 6. 

We do so by comparing the singular-value spectra of matrices drawn according to COE$^M_{N,K}$ and $\mathcal G^{\text{sym}}_{N,K}(0,1/M)$, respectively. 
Since both distributions COE$^M_{N,K}$ and $\mathcal G^{\text{sym}}_{N,K}(0,1/M)$ over complex, symmetric $N\times N$ matrices are invariant under conjugation with $V \cdot V^T$ for any $N\times N$ unitary matrix $V$, the probability of drawing a particular matrix $C$ from these distributions depends only on the singular values of the matrix $C$. 
Consequently, the distribution of singular values captures the essence of both distributions alike. 
Let $P(r)$ denote this distribution, that is, the distribution over singular values $r$ of a matrix $C$ drawn either from COE$^M_{N,K}$ and $\mathcal G^{\text{sym}}_{N,K}(0,1/M)$.

\begin{figure*}
    \centering
    \includegraphics[width = \linewidth]{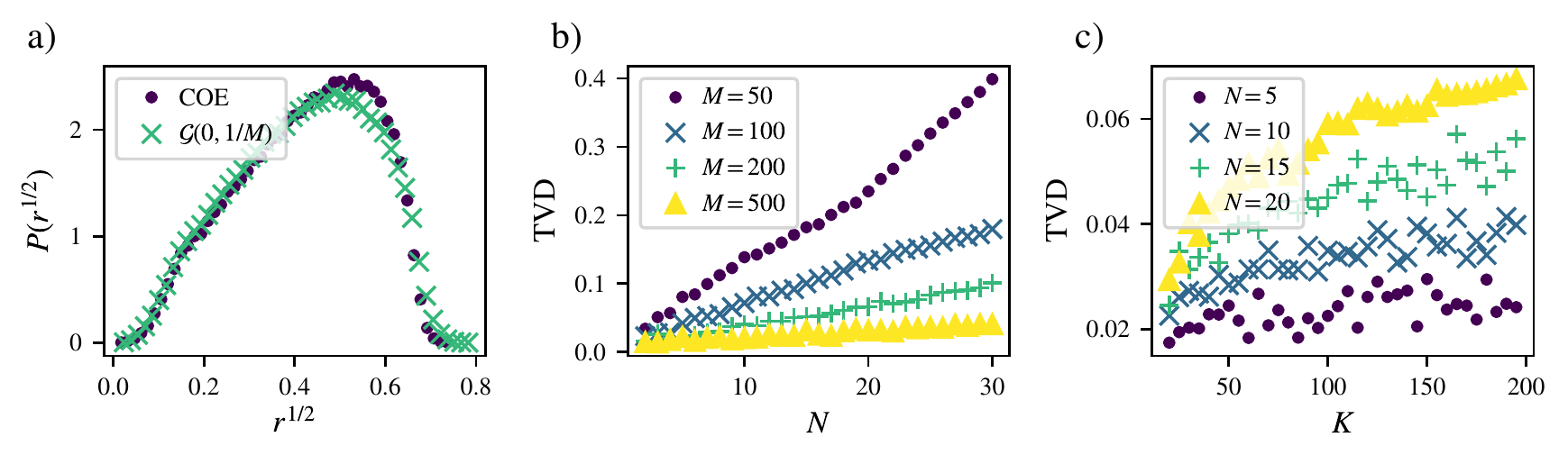}
    \caption{\raggedright
    \textbf{Numerical evidence that the ensembles COE$^M_{N,K}$ and $\mathcal G^{\text{sym}}_{N,K}(0,1/M)$ converge in total-variation distance for any $K \geq N$ so long as $N \in o(\sqrt{M})$.} 
    a) The singular-value spectra of COE$^M_{N,K}$ and $\mathcal G^{\text{sym}}_{N,K}(0,1/M)$ for $M=200$, $K=200$, and $N=10$.
    b) Total variation distance between singular-value spectra of COE$^M_{N,K}$ and $\mathcal  G^{\text{sym}}_{N,K}(0,1/M)$ for different $M=K$ as a function of $N$.
    c) Total variation distance between singular-value spectra of COE$^M_{N,K}$ and $\mathcal G^{\text{sym}}_{N,K}(0,1/M)$ for $M=200$ and different $N$ as a function of $K$.}
    \label{fig:hiding_numerics}
\end{figure*}

In Fig.~S1(a), we show the finite approximation to the distribution $P(r)$ for both ensembles under consideration for fixed values of $M,K,N$. 
While the distributions differ (as expected for any finite matrix size), they are already very close to each other for reasonably small matrices. 
In Figs.~S1(b) and (c), we then further investigate the scaling of the total-variation distance between finite-bin approximations of $P(r)$ for COE$^M_{N,K}$ and $\mathcal G^{\text{sym}}_{N,K}(0,1/M)$ in the size of the sub-matrices. 
In Fig.~S1(b), we consider the scaling of the total-variation distance in the short side $N$ of $N \times M$ sub-matrices, i.e., for the second regime where $K=M$. 
As expected, the total-variation distance increases with $N$ but decreases as the relative size of $N$ to $M$ decreases, too. 
This provides evidence that the rigorous result about the asymptotic convergence of COE$_{N,K}$ and $\mathcal G^{\text{sym}}_{N,K}(0,1/M)$ for $K=M$ due to Jiang~\cite{jiang_entries_2009} can be strengthened to an inverse polynomial total-variation distance bound (Conjecture 6.2).
Finally, in Fig.~S1(c), we show that the size of the long side $K$ of the sub-matrices does not significantly affect the total-variation distance in the regime of $N \ll M$ (the collision-free regime. 
This constitute evidence that the value of $K \geq N$ does not make a significant difference to the closeness of the distributions COE$_{N,K}$ and $\mathcal G^{\text{sym}}_{N,K}(0,1/M)$ of symmetric matrix products. 

To summarize this section, we have formulated an interesting conjecture regarding the distribution of symmetric products of sub-matrices of Haar-random unitaries.
In the main text, we argued that this conjecture captures the hiding property for Gaussian boson sampling. 
Here, we have provided analytical evidence for the conjecture in the two extremal regimes of $K=N$ (where we know both parts to be true) and $K=M$ (where we know part 1.\ to be true). 
We then provided numerical evidence for an inverse polynomial total-variation distance bound for any value of $K$ such that $N \leq K \leq M$. 

Let us note that -- as in the case of standard boson sampling -- our conjecture does not apply to the case in which $N \in \Omega(\sqrt M)$. Indeed, for the case of $N = K \in \Omega(\sqrt M)$ Ref.~\cite{jiang_entries_2009} shows that $\mathcal H_{N,N}^M$ and $\mathcal G_{N,N}(0,1/M)$ are far from each other in total-variation distance.
This indicates that the statement of our conjecture does not hold in this case since there is no `short side' of $U_{\mathbf n,1_K}$.


\section*{Average-case hardness of computing GBS output probabilities} \label{sec_apdx_avgcaseGBS}

In this section, we show average-case hardness of computing GBS output probabilities. As explained in the main text, this amounts to showing that the following problem is $\#\mathsf{P}$-hard.

\begin{problem}{$(\delta,\epsilon)$-Squared-Hafnians-of-Gaussians}
Input & A matrix $XX^T$ with $X  \sim \mathcal G_{N,K}(0,1/{M})$.\\
Output & ${|\mathrm{Haf}{(XX^T)}|}^2$ to additive error $\epsilon$, with probability $\geq \delta$ over the distribution $\mathcal G_{N,K}(0,1/{M})$.
\end{problem}

The proof will proceed in two steps: First, we will show that an oracle for the \textsc{$(\delta,\epsilon)$-Squared-Hafnians-of-Gaussians} problem allows one to approximate $\left|\mathrm{Haf}(YY^T )  \right|^2$ for arbitrary  $Y \in \mathbb{C}^{2N \times 2K}$. This first part of the proof constitutes the worst-to-average-case reduction. Second, we will show that approximating $\left|\mathrm{Haf}(YY^T )  \right|^2$ for arbitrary  $Y \in \mathbb{C}^{2N \times 2K}$ is actually $\#\mathsf{P}$-hard in the worst-case. We show this by reducing the task of approximating the permanent of an arbitrary complex $N\times N$ matrix to the task of approximating $\left|\mathrm{Haf}(YY^T )  \right|^2$.

\subsection*{Worst-case hardness} \label{ssec_worstcaseGBS}
Consider the following problem:
\begin{problem}{$\epsilon$-Squared-Hafnians}
\label{prob:Hafnian wc}
Input & A matrix $YY^T$ with $Y \in \mathbb{C}^{N \times K}$ for $K \in \mathbb N$, $N \in 2\mathbb{N}$ such that the entries of $Y$ are of the form $(x +iy)/\sqrt{M}$ for $|{x}|,
|{y}|$ some $O(1)$-bounded integers and additive-error tolerance $\epsilon>0$.
\\
Output & An estimate $h$ s.t. ${\left|h - \left|\mathrm{Haf}( YY^T )\right|^2 \right|} \leq \epsilon$.
\end{problem}
We prove the following Lemma.
\begin{lemma}\label{worst_case_squared_Hafnian}
The problem \textsc{$\epsilon$-Squared-Hafnians} is worst-case $\#\mathsf{P}$-hard for any additive error $\epsilon \leq 1/(2M^N)$.
\end{lemma}

\begin{proof}
Without loss of generality, we restrict to $N \leq K$. 
We begin the proof by noting that the permanent of any square matrix $G$ can be expressed as the Hafnian of a corresponding block matrix twice the size of $G$~\cite{Kruse2019},
\[
\text{Per}(G) = \mathrm{Haf}\left[\begin{pmatrix}
0 & G\\
G^T & 0
\end{pmatrix} \right].
\]
Hence, computing the squared permanent of any complex  $N/2 \times N/2$ matrix  $G\in \mathbb C^{N/2 \times N/2}$ reduces to computing the squared Hafnian of a corresponding block matrix 
\begin{equation}
B(G)=  \begin{pmatrix}
0 & G\\
G^T & 0
\end{pmatrix}.
\end{equation}
Computing the squared permanent exactly is known to be worst-case $\#\mathsf{P}$-hard even over $0/1$-matrices  \cite{Valiant1979,Aaronson2013}.

Next we note that any matrix $B(G)$ for $G \in \mathbb C^{N/2 \times N/2}$ can be decomposed as $XX^T$ in terms of some complex matrix $X \in \mathbb{C}^{N \times K}$. 
Indeed the block matrix $B(G)$ is a complex, symmetric matrix, so we can decompose it using the Takagi decomposition as $WDW^T$, where $W \in U(N)$ is a unitary matrix and $D \in \mathbb R^{N \times N}$ is a nonnegative diagonal  matrix. 
We now define $X'=(WD^{1/2})$ and $X$ by appending $(K-N)$ all-$0$-columns to $X'$. 
This gives rise to a decomposition of $B(G) = XX^T$ with $X \in \mathbb C^{N \times K}$. Hence it is $\#\mathsf{P}$-hard to exactly compute the Hafnian of matrices of the form $XX^T$ in the worst case. Additionally, since the Hafnian is a continuous function, we can compute $\mathrm{Haf}(XX^T)$ to an arbitrary level of precision by considering $\mathrm{Haf}(YY^T)$ with the entries of $Y$ being of the form $x+iy$, with $x$ and $y$ integers (by suitably rescaling the entries of the matrix). Finally, we note that by normalization we can assume that the entries of the matrix $Y$ are of the form $(x +iy)/\sqrt{M}$ with $x$ and $y$ $O(1)$ bounded integers. Then the squared Hafnian of $YY^T$ is an integer multiple of $1/M^N$. Therefore, computing the Hafnian of $YY^T$ up to additive error of $1/(2M^N)$ serves to compute the squared Hafnian of $B(G)$ exactly, which is $\#\mathsf{P}$-hard.
This concludes the proof.

The proof holds equally for $N \in \mathsf{poly}(K)$: in this case we embed a square matrix in $\mathbb C^{K \times K}$ and append $0$ rows instead of columns. 
\end{proof}

\subsection*{Worst-to-average equivalence}

We now prove the average-case hardness of computing GBS output probabilities. That is, we prove the following Lemma:
\begin{thm}[Theorem 3 restated] 
The \textsc{$(\delta,\epsilon)$-Squared-Hafnians-of-Gaussians} problem is $\#\mathsf{P}$-hard under $\mathsf{PH}$ reductions for any $\epsilon \leq O\left(\exp[-6N\log N - \Omega(N)]\right)$ and any constant $\delta > 3/4$.
\end{thm}
We first sketch the proof idea and elaborate on the technique used.
The overall idea is to give a worst-to-average-case reduction from the problem \textsc{$\epsilon$-Squared-Hafnians} to the problem \textsc{$(\delta,\epsilon)$-Squared-Hafnians-of-Gaussians}.
The worst-case $\#\mathsf{P}$-hardness of problem \textsc{$\epsilon$-Squared-Hafnians} has already been established.
 
We use the same technique as Refs.~\cite{Aaronson2013,bouland2021noise} to establish this reduction.
Assume that we are given an oracle $O$ that solves \textsc{$(\delta,\epsilon)$-Squared-Hafnians-of-Gaussians}, meaning that with probability at least $\delta$ over the input $X$, it outputs a squared Hafnian of $XX^T$ to additive error $\epsilon$.
The rest of the time, it may output an incorrect value, with no guarantees whatsoever on how close the output is to the desired output.
In the following, we will show how to use the oracle $O$ to obtain the squared Hafnian of an arbitrary worst-case matrix $YY^T$ with high probability (this latter probability is over the choice of the random variables instantiated in the algorithm).

The key idea is that for $X \in \mathbb{C}^{N \times K}$, the quantity $|{\mathrm{Haf}{(XX^T)}}|^2$ is a degree $2N$ polynomial over the entries of the matrix $X$.
This allows for the use of polynomial interpolation to recover the squared Hafnian of an arbitrary worst-case matrix $YY^T$.
An important technique we use in this proof is the robust Berlekamp-Welch algorithm due to Ref.~\cite{bouland2021noise}, which is important for polynomial interpolation over $\mathbb{R}$ as opposed to a finite field.
Polynomial interpolation over the reals is a technique often used for the problem of average-case hardness of computing output probabilities of random quantum circuits~\cite{Bouland2019,Movassagh2019}.
The Berlekamp-Welch algorithm cannot be used as is for the reals, and therefore, recent works~\cite{Bouland2019,Movassagh2019} use techniques like Lagrange interpolation.
The new robust Berlekamp-Welch algorithm of Ref.~\cite{bouland2021noise} allows for improved robustness of the worst-to-average-case reduction. 

As an example, in the context of random quantum circuits over $n$ qubits and $m$ gates, Lagrange interpolation can only give average-case $\#\mathsf{P}$-hardness of computing output probabilities to error $2^{-O(m^3)}$ rather than the $O(2^{-n})$ that suffices for proving the hardness of approximate sampling (see \cite{Aaronson2013, Movassagh2019}).
The modified Berlekamp-Welch algorithm of Ref.~\cite{bouland2021noise}, which is boosted with an $\mathsf{NP}$ oracle, can sidestep the need for Lagrange interpolation and obtain average-case $\#\mathsf{P}$-hardness with $2^{-O(m \log m)}$ error (see also, the recent work of Kondo \emph{et al}.~\cite{kondo2021fine-grained} which also obtains this robustness error).

\begin{thm}[Robust Berlekamp-Welch algorithm~\cite{bouland2021noise}] \label{thm_BW}
Let $p$ be a univariate polynomial of degree $d$ over the reals.
Suppose that we have $k \geq 100d^2$ points $(x_i,y_i)$, with $\{x_i\}$ uniformly spaced in the interval $[0,\kappa]$ and obeying the promise
\[
 \Pr[|{y_i - p(x_i)}| \geq \Delta] \leq \eta < \frac{1}{4} .
\]
Then there is a $\mathsf{P}^\mathsf{NP}$ algorithm that can estimate $p(1)$ to additive error $\Delta \exp[d \log \kappa^{-1} + O(d)]$ with probability at least 2/3.
\end{thm}

\begin{proof}[Proof of Theorem 3]
The polynomial interpolation procedure is as follows.
Let $X(t)$ be the matrix obtained by drawing a random $X \sim \mathcal G_{N,K}(0,1/M)$ and setting 
\[
X(t) := (1-t) X + tY,
\]
where $Y$ is the matrix corresponding to the worst-case instance.
Now, the quantity 
\[
p(t) := |{\mathrm{Haf} (X(t)X^T(t))}|^2 
\]
is a polynomial of degree $2N$ over the entries of $X(t)$, and consequently, over $t$ itself.
For $t$ close to 0, $X(t)$ is close to Gaussian distributed, while when $t$ is close to 1, the distribution is close to being deterministic.
We select $k$ points in the range $[0,\kappa]$ and query the oracle $O$ for the value of $p(t)$ for these points.
By the promise, the oracle outputs the correct value of $p(t)$ for most values of $t$ with high probability.
Conditioned on this event, the robust Berlekamp-Welch algorithm stated in Theorem 9 allows one to reconstruct the polynomial in the second level of the polynomial hierarchy.
The polynomial can then be evaluated at the point $t=1$ to obtain an estimate of the squared Hafnian of the worst-case matrix $YY^T$.

We now check that the conditions of Theorem 9 are met.
We say that a call to the oracle $O$ is successful if it outputs the squared Hafnian of a matrix to additive error $\epsilon$.
By assumption, for $X$ drawn at random from $\mathcal G_{N,K}(0,1/M)$, the oracle is successful with probability at least $\delta$.
Note however that the matrix $X(t) = (1-t) X + tY$ is not exactly distributed according to $\mathcal G_{N,K}(0,1/M)$. Instead, for small $t$, due to the rescaling by $(1-t)$ and the shift by $tY$, $X(t)$ is distributed according to a slightly different distribution $\mathcal{G}'$.
If we query the oracle for the value of $p(t)$ with matrices drawn from this different distribution $\mathcal{G}'$, the probability of success can, in the worst case, decrease.
By definition, the success probability can decrease at most by the variation distance between the two distributions $\mathcal G_{N,K}(0,1/M)$ and $\mathcal{G}'$, which is $O(t\max(N,K)^2)$.
Therefore, for $K\geq N$, the probability of success is at least $\delta - O(\kappa K^2)$.
We choose $\kappa$ to be $O(c/K^2)$ with some small enough $c$ so that the success probability is at least $ \delta - O(c) > 3/4$.
This ensures that the conditions of the theorem are met.

We finally conclude by examining the additive error to which we can compute, using the $\mathsf{BPP}^{\mathsf{NP}}$ reduction, the squared Hafnian of the worst-case matrix $YY^T$.
If the additive error for successful queries to the oracle is at most $\epsilon$, Theorem 9 implies that the error in computing $p(1)$ is $\epsilon \exp[d \log \kappa^{-1} + O(d)]$.
Plugging in $d=2N$ and $\kappa = c/N^2$, we get the total additive error in estimating $p(1)$ to be $\epsilon \exp[4N \log N + O(N)]$.
Finally, we note that the squared Hafnian is shown to be worst-case hard for additive error $O(1/M^N)$.
Therefore, we make the choice
\begin{equation}
\epsilon \exp[4N \log N + O(N)] 
= O\left(\frac{1}{M^N}\right),
\end{equation}
or
\[
\epsilon =\, O\left(\exp[-4N \log N - \Omega(N) - 2N\log N ] \right)
\\ =\, O\left(\exp[-6N \log N - \Omega(N)] \right),\nonumber
\]
where we have assumed $M=\Theta(N^2)$. 
This choice ensures that we can, with probability at least 2/3, compute the squared Hafnian of an arbitrary matrix with bounded entries of the form $YY^T$ to additive error $O(1/M^N)$. As shown in Lemma 7, this task is $\#\mathsf{P}$-hard.
This completes our proof.
\end{proof}

\section*{Average-case hardness of computing noisy GBS output probabilities} \label{sec_apdx_avgcasenoisyGBS}
We argue here that computing the output probabilities for a \emph{noisy} random GBS experiment is $\#\mathsf{P}$-hard on average.
That is, we show the following lemma.
\begin{lemma} \label{lem_worstcasehardnoisy}
There exists a polynomial $p(N)$ and a loss threshold $\eta_*$ such that \textsc{$(\epsilon,\eta)$-NoisyGBS-Probability} with $\eta \leq \eta_*$, $\delta > 3/4$, and $\epsilon \leq 2^{-p(N)}$ is $\#\mathsf{P}$-hard under $\mathsf{PH}$ reductions.
\end{lemma}

\begin{proof}
For worst-case hardness despite the presence of noise, we follow the proof technique in Refs.~\cite{fujii2016noise,bouland2021noise}.
At a high level, the worst-case hardness follows from the error-detection property of the system.
In particular, the error-detection property implies that as long as the noise $\eta$ is smaller than a certain threshold $\eta_*$, there is a fixed outcome on a subset of the modes, say $\mathbf{m}$, such that conditioned on this outcome, the probability distribution on the rest of the modes is exponentially close to the target noiseless distribution.
In other words, we have
\[
 \left|{\Pr_\text{noisy}[\mathbf{n}|\mathbf{m}] - \Pr_\text{ideal}[\mathbf{n}]}\right| \leq 2^{-\mathsf{poly}(N)}
\]
for any desired polynomial on the right hand side.
Since $\Pr_\text{ideal}[\mathbf{n}]$ is $\#\mathsf{P}$-hard to approximate in the worst case by virtue of Lemma 7, so is computing the conditional probability
\[
 \Pr_\text{noisy}[\mathbf{n}|\mathbf{m}] = \frac{\Pr_\text{noisy}[\mathbf{n},\mathbf{m}]}{\Pr_\text{noisy}[\mathbf{m}]}.
\]
The denominator here is the probability of seeing the outcome $\mathbf{m}$, which flags the no-error event.
The probability of this can be exponentially small, and satisfies~\cite{fujii2016noise,bouland2021noise}
\[
\left|{\Pr_\text{noisy}[\mathbf{m}] - (1-\eta)^{O(Md)}}\right| \leq \Pr_\text{noisy}[\mathbf{m}] 2^{-\mathsf{poly}(N)},
\]
where $\eta$ is the maximum noise parameter as defined earlier in the main text.
In other words, for an error-detected circuit, the probability that the outcome on the subset of heralding modes is in the state $\mathbf{m}$ is exponentially close to the probability that no error occurred, which is given by $(1-\eta)^{O(Md)}$.

Therefore, approximating $\Pr_\text{noisy}[\mathbf{n},\mathbf{m}]$ is also $\#\mathsf{P}$-hard:
\begin{gather}
\left|{\Pr_\text{noisy}[\mathbf{n},\mathbf{m}] - \Pr_\text{noisy}[\mathbf{n}|\mathbf{m}] (1-\eta)^{O(Md)}}\right| \leq \Pr_\text{noisy}[\mathbf{n},\mathbf{m}] 2^{-\mathsf{poly}(N)}
\\ \Rightarrow \left|{\Pr_\text{noisy}[\mathbf{n},\mathbf{m}] - \Pr_\text{ideal}[\mathbf{n}] (1-\eta)^{O(Md)}}\right| \leq \Pr_\text{noisy}[\mathbf{n},\mathbf{m}] 2^{-\mathsf{poly}(N)} \nonumber +  2^{-\mathsf{poly}(N)}. 
\end{gather}
Since computing $\Pr_\text{ideal}[\mathbf{n}]$ to additive error $\pm O(2^{-\mathsf{poly}(N)})$ is $\#\mathsf{P}$-hard, so is computing $\Pr_\text{noisy}[\mathbf{n},\mathbf{m}]$ to additive error $O(2^{-\mathsf{poly}(N})(1-\eta)^{O(Md)}$.
A similar analysis in Ref.~\cite{bouland2021noise} shows that it is $\mathsf{coC_{=}P}$-hard to compute a noisy probability in the worst case to additive error $2^{-O(m\log m)}$ in the context of RCS.
This proves the worst-case hardness.

For the worst-to-average-case reduction, we again use the technique of polynomial interpolation in conjunction with a robust Berlekamp-Welch algorithm.
We observe that any noisy output probability for a local noise model can still be written as a polynomial in the gate entries of the circuit, using the Feynman sum-over-paths idea.
As before, we perform interpolation from a random instance from the ensemble to the worst-case-hard instance.
This is achieved now using the Cayley path interpolation technique of Ref.~\cite{Movassagh2019} instead of the direct interpolation between two matrices. This is because the noisy output probability is no longer a simple function of only the linear-optical unitary (like the Hafnian), but is also a function of the circuit implementation.
The full interpolation involves interpolating every gate of a circuit implementation from the average-case instance $A_i$ to the worst-case instance $W_i$ along the Cayley path
\[
C_i(t) = \left(t \mathds{1} + (2-t)A_i W_i^{-1}\right) \left( (2-t)\mathds{1} + t A_i W_i^{-1} \right)^{-1} \cdot W_i,
\]
which satisfies $C_i(0)=A_i$ and $C_i(1)=W_i$.
Using this interpolation and the fact that any local noise can be ``purified'' gate-wise by introducing ancillary systems of finite dimension, we can again write the noisy probability $\Pr_\text{noisy}[\mathbf{n},\mathbf{m}][t]$ as a polynomial in $t$.
The rest of the proof follows from before.
\end{proof}

\section*{Average-case hardness of computing noisy probabilities in high-dimensional GBS} \label{sec_apdx_avgcasenoisyrestrictedGBS}
For the worst-case hardness of computing noisy probabilities of the high-dimensional GBS architecture, we mainly use the previous results on error-detection of noise.
The additional ingredient used is the fact that a constant-depth linear-optical architecture in two dimensions (and higher) has been shown by Brod~\cite{Brod2015} to be hard to exactly sample from.

The proof of Ref.~\cite{Brod2015} uses post-selection to argue for exact sampling hardness.
Note that the post-selection result does not, by itself, imply the $\#\mathsf{P}$-hardness of computing output probabilities: it implies the $\mathsf{PP}$-hardness of strong simulation, which involves computing both the output probabilities and the marginals.
However, we note that the post-selection proof can often be ``opened up'' in order to directly argue about the hardness of computing output probabilities.
This is done by giving an amplitude-preserving reduction from a $\mathsf{BQP}$ circuit to the circuit family in question (here, high-dimensional GBS).
Since computing output amplitudes of $\mathsf{BQP}$ circuits is $\#\mathsf{P}$-hard, so is computing that of the circuit family in question.
Using the results from earlier, so is computing the \emph{noisy} output probability in the worst case for an error-detected circuit as long as the noise level is smaller than some (constant) threshold $\eta_*$.

The average case hardness again essentially follows by observing that there is a polynomial structure in the output probability, to prove Theorem 5.
We again use the Cayley technique of Ref.~\cite{Movassagh2019} to set up the polynomial interpolation in this case, and use results from Ref.~\cite{bouland2021noise} to strengthen it, such as using a variable rescaling and applying a robust version of the Berlekamp-Welch algorithm (Theorem 9).

\section*{Total photon number distribution}~\label{app:total-photon-number}

For pure state GBS, the total photon number distribution can be obtained efficiently by simply convolving the photon number distributions of the individual modes going into the interferometer~\cite{gupt2019walrus}. In the case where $M$ identical squeezed states (with squeezing parameter $r$) are sent into an interferometer and undergo uniform loss by transmission parameter $\eta$, the probability of obtaining $n$ photons is given by 
\[
		\Pr(n) = \begin{cases}
			\eta ^n \binom{\frac{M}{2}+\frac{n}{2}-1}{\frac{n}{2}} \mathrm{sech}^M r \tanh^n (r)
            _2F_1\left(\frac{n}{2}+\frac{1}{2},\frac{M}{2}+\frac{n}{2};\frac{1}{2};(1-\eta)^2 \tanh ^2 r \right) \text{ if $n$ is even,} \\
			(1-\eta) (n+1) \eta ^n \binom{(M+n-1)/2}{(n+1)/2} \mathrm{sech}^M r \times \\ \tanh ^{n+1} (r)  _2F_1\left(\frac{n+2}{2},\frac{1}{2}
			(M+n+1);\frac{3}{2};(1-\eta)^2 \tanh ^2 r \right) \text{if odd.}
		\end{cases}
\]
where $_2F_1(a,b,c;z)$ is a hypergeometric function.
This equation reduces to the well-known lossless limit~\cite{hamilton2017gaussian} when $\eta \to 1$, since in that case  $_2F_1\left(\frac{n}{2}+\frac{1}{2},\frac{M}{2}+\frac{n}{2};\frac{1}{2};0\right) =1$ and the probabilities for all odd photon numbers become zero since they are proportional to $1-\eta$.
This distribution has the following moments
\[
\mathbb{E}(n) =  \eta M \sinh^2 r, \\
\text{Var}(n) = \eta  M \sinh ^2 r [1+ \eta  (1+2 \sinh^2  r)] .
\]
Note that even if the losses are not uniform, one can still calculate in polynomial time the moments of the random variable $n$~\cite{dodonov1994photon,gupt2019walrus}.

\end{document}